\documentclass[12pt, english]{article}
\usepackage{lmodern}
\usepackage{geometry}
\usepackage{amsmath}
\usepackage{amsthm}
\usepackage{multirow}
\usepackage{amssymb}
\usepackage{bbm}
\newtheorem{theorem}{Theorem}
\newtheorem{proposition}{Proposition}
\usepackage{natbib}
\usepackage{graphicx}
\usepackage{tabularx,booktabs}
\usepackage{setspace}
\usepackage{caption}
\usepackage[bottom,hang,flushmargin]{footmisc}
\usepackage{datetime}
\usepackage{url}
\usepackage[most]{tcolorbox}
\newtcolorbox{defbox}{colback=gray!8, colframe=black!60, boxrule=0.5pt, arc=2pt,
  left=8pt, right=8pt, top=5pt, bottom=5pt, before skip=10pt, after skip=10pt}

\begin{document}

\title{{\Large\textbf{Diversity, Efficiency, and Incentives\\
in Political Prediction Markets}\vspace{0.25in}\thanks{We thank Kalshi, PredictIt/Aristotle Inc., Paul Rhode, Andreas Mastrosavvas, Mick Bransfield, and Dan Rosenheck (The Economist) for providing data used in the analysis. The authors are responsible for any errors.}}}

\author{\textbf{Harry Crane}\thanks{Department of Statistics, Rutgers University. Email: hcrane@stat.rutgers.edu}  \\ \small{Rutgers University}  \and \textbf{Koleman Strumpf}\thanks{Department of Economics, Wake Forest University. Email: strumpks@wfu.edu }
\\ \small{Wake Forest University}
}

\date{Working Paper\\
\monthname\ \number\year}
\maketitle

\begin{abstract}
Critics of political prediction markets often argue that prices are biased because traders are not representative of voters. In this paper, we test this claim by analyzing comprehensive data on trader demographics from political betting markets spanning more than 100 years. Our analysis yields two core findings: (1) Across the markets we studied, traders were not representative of the voters in any of them, though they were non-representative in different ways. (2) Despite the lack of representativeness among traders, market prices were not biased by the political preferences of their participants. Specifically, for the 2024 election, we analyzed anonymized trade-level data from Kalshi, a CFTC-regulated exchange, the largest regulated political prediction market to date in the United States. In 2024, Kalshi's election-market traders were 86\% male and half were under 30. Despite the skewed demographics, Kalshi outperformed leading forecasters and poll aggregators decisively. In its first month of listing election contracts, Kalshi's prices were more accurate than every poll-based forecast in the national race, the seven battleground states, and nineteen Senate races, and no poll-based forecast would have been profitable to trade against them in any group of presidential contests. Furthermore, the claim of Republican bias based on trader non-representativeness is disputed by the data, as our analysis finds that its volume came from areas that voted 28\% Republican, while 60\% of its volume went to Republican positions. Traders with just a single trade, whose behavior most closely resembles political expression, were the most biased by their inferred preferences and placed 73\% of their volume on Harris. The association between a trader's location and the direction of the trader's positions declines with trading activity and is zero among traders with more than 100 trades, who account for most of the trading volume and are primarily responsible for influencing market prices. Our analysis shows that prediction market prices can be accurate even with unrepresentative traders when the forecasts of profit-motivated traders do not depend on their demographics. Data from two other historical prediction markets corroborate our main findings. In New York's 1900 election market, bettors were older than the population and disproportionately drawn from finance. On PredictIt in 2016, which limited each trader to \$850 per contract, positions followed registered party more closely, as our theory predicts when the market restrictions do not allow profit-motivated traders to set the price. Overall, we find that prediction market efficiency is driven by financial incentives over demographic diversity or representativeness.
\end{abstract}

\bigskip
\noindent\textbf{JEL Classification}: D12, D72, D81, D91, G14

\noindent\textbf{Keywords:} prediction markets, election forecasting, diversity, wisdom of crowds, Brier score

\pagebreak{}

\begin{spacing}{1.5}

\section{Introduction}\label{section:introduction}

Through the fall of 2024, U.S. election prediction markets priced Donald Trump 5 to 10 percentage points above the leading poll-based forecasting models.  The disagreement prompted criticism from academics, the media, and legacy polling institutions. \cite{sonnenfeld2024} told readers not to trust the political prediction markets, describing them as thin, manipulable, and open only to a subsection of the population. In \textit{Fortune}, a business-law professor observed that on Polymarket ``you don't have a particularly large or diverse crowd of people,'' and a public-affairs consultant that the markets ``attract a specific kind of participant: people interested in betting who might not reflect the broader population'' \citep{lake2024}. A Benzinga analysis reasoned that because twice as many men as women hold cryptocurrency and men favored Trump, the crypto-funded Polymarket was sampling the wrong electorate, and that ``efficiency in markets relies on diversity among participants'' \citep{buckfire2024}. Nate Silver, whose model would finish the campaign at 50--50 and who was at the time a paid adviser to Polymarket, wrote that the platform's ``crypto-adjacent user base has become more Trump-leaning'' \citep{silver2024benzinga}, and a week before the election told CNBC that the markets' prices were ``calibrated based on vibe and chatter'' and ``can get a little carried away'' \citep{silver2024cnbc}. Rajiv Sethi, whose work with David Rothschild had documented that Intrade's traders were ``predominantly young and male'' and that its prices skewed toward the best-funded rather than the best-informed \citep{rothschild-sethi2016}, suggested to the \textit{Wall Street Journal} that the Trump-skewed prices on Polymarket were likely the result of market manipulation \citep{osipovich2024}, a suggestion he withdrew two weeks later once the trader's identity and motive were known \citep{sethi2024limits}. After the election, Better Markets pressed the same point, calling bettors ``mostly young males'' who ``are not and never will be a statistically valid sample of likely voters'' \citep{kelleher2024}. 

This paper tests that criticism with anonymized trade-level data from Kalshi, a CFTC-regulated exchange and the largest regulated political prediction market to date in the United States. Overall, the critics' premise is correct: Kalshi's election-market traders were 86\% male and half were under 30.  But they wrongly asserted that this lack of representativeness implied inaccuracy.  In its first month of listing election contracts, Kalshi's prices were more accurate than every poll-based forecast in the national race, the seven battleground states, and nineteen Senate races, and no poll-based forecast would have been profitable to trade against Kalshi's prices in any group of presidential contests. They also had the direction of the political bias backwards.  Although Kalshi's prices favored Trump, its traders tended to favor Harris. Kalshi's volume came from areas that voted only 28\% Republican, and the traders with a single trade, whose behavior most resembles political expression, placed 73\% of their volume on Harris. The traders who set prices, those with more than 100 trades and 93\% of volume, took positions whose direction was uncorrelated to the politics of their home zip codes. 

The above criticisms focus on attributes of the people involved in prediction market activity rather than the mechanism underlying prediction market price formation. The 2024 arguments cited above suggest that market prices are biased because market participants are not representative of the electorate. Underlying those arguments is the implicit assumption that traders take positions which reflect their political identities, so that the price aggregates the trader population's politics rather than its information. Whether the absence of demographic diversity damages prediction market prices has not been rigorously tested. We test this argument with data from several prediction markets spanning more than a century of elections. Throughout, we call a market efficient when its prices are accurate forecasts and are not systematically biased toward the candidate its traders prefer.

The markets we study are Kalshi, a CFTC-regulated exchange that listed election contracts one month before the November 2024 election, PredictIt in 2016, and the New York betting markets of 1884 to 1940 \citep{rhode-strumpf2004, rhode-strumpf2013}. We compare their prices with those of Polymarket, the Betfair betting exchange, and UK bookmakers, and with forecasts from the leading poll-based models. Section~\ref{section:data} describes the data.

Our main findings are as follows.

\textbf{First, the premise of the diversity critique is correct. Political prediction market participants are not representative of the electorate.} On Kalshi in 2024, traders were far younger and more often male than the electorate. Kalshi's traders were spread across the country in rough proportion to population. Section~\ref{section:demographics} reports the comparison. The same holds in the other two markets we study, each unrepresentative in its own way. On PredictIt in 2016, 92\% of the traders in the presidential market with a recorded sex were men and about 42\% were under 30, against an electorate that was 52\% female and 19\% under 30.\footnote{National exit poll conducted by Edison Research for the National Election Pool, November 2016. Sex is recorded for 55\% of the PredictIt traders. Age in the PredictIt data is age in July 2025, which we convert to age in November 2016.} In New York's 1900 election market, all 80 bettors we match to the Census were men, 44\% were bankers or brokers against 0.6\% of New York City's men at the time, and their median age was 41 against 22 for the city's males. 

\textbf{Second, despite the above demographic discrepancy, market prices were more accurate than every poll-based model, nationally and in the state and Senate races.} Scoring eight forecasters by daily Brier score over the common window, the four uncapped market-based forecasters, Kalshi among them, score between 0.166 and 0.185 in the national race, all better than the three poll-based models (Silver Bulletin 0.241, FiveThirtyEight and The Economist 0.252), and PredictIt, the exchange that caps each trader at \$850 per contract, falls between the two groups (0.219). In the down-ballot contests, Kalshi scores better than both poll-based models that published state forecasts in every one of the seven battleground states, better than The Economist and FiveThirtyEight in nineteen contested Senate races, and close to Polymarket in every contest group. In the nine governor races, which only the two exchanges forecast, Kalshi scores best. The exception is the House. In the twenty-four district contracts Kalshi opened six days before the election, which drew under \$250,000 in total, FiveThirtyEight's model scores 0.192 to Kalshi's 0.241, and a trader who believed the model would have profited against the market. 
Under the market-based scoring introduced by \cite{crane2018} and developed further by \cite{crane-vinson2020} based on 2018 and 2020 election markets, a trader who believed any poll-based model and traded against Kalshi's prices would have lost money in every presidential contest group and roughly broken even in the Senate races. 

\textbf{Third, trading behavior follows financial incentives, not demographics.} The selection gap itself runs the wrong way for the critics' argument. Kalshi's volume came from areas that voted only 28\% Republican, so a market that priced its traders' politics would have been biased toward Harris, not Trump. Four further pieces of evidence point the same way. (i) Traders with more than 5,000 trades account for 79\% of volume, placed 64\% of their main-contract volume long Trump, and had the highest returns of any cohort. Single-trade users placed 73\% of their volume long Harris and had the lowest return. (ii) Attributing each trader fractionally to the parties by the 2024 vote of the trader's location, so that a trader in an 80--20 Trump location counts as 0.8 Republican and 0.2 Democrat, the dollar-weighted partisan gap in trade direction across all election markets declines with trade intensity, from about 7 percentage points among casual traders to zero among traders with more than 100 trades.  (iii) Trading profit is unrelated to the demographic dimensions on which the platform is least diverse. Male and female traders are profitable at approximately the same rates.  Among traders with more than 100 trades, four-fifths of volume comes from two-sided traders, who hold positions on both sides of a contract. These traders paid prices that proved fair on average, and the direction of their positions is unrelated to the politics of their location. The platform's male majority did not tilt it toward Trump either. Men voted for Trump by 55 to 43, so a market in which traders bought the candidate they preferred would show male traders on the Republican side. Among traders with ten or fewer trades, male traders instead put a smaller share of their volume on the Republican side than female traders did, and among more active traders the two sexes do not differ. (iv) In user-level regressions, demographics and local partisanship together explain under 8\% of the variation in trading direction even with trade-intensity controls, and the within-cohort partisanship gradient falls from roughly 0.5 among casual traders to 0.08 among traders with more than 100 trades.

We extend the analysis to PredictIt's 2016 presidential market. The qualitative pattern is the same, but the preference signal is larger. Registered Democrats placed 36\% of their shares long Trump and registered Republicans 58\%. On a small, politically engaged exchange with low investment limits, the incentives mechanism is less of a factor in trading behavior. The contrast between the two markets is consistent with the incentives account. Market prices are driven by the share of volume held by profit-motivated traders, and PredictIt's investment cap prevents that share from growing large enough to set the price. \cite{zitzewitz2026} studies the price consequences of the trading cap, and historical evidence supports the same conclusion. The New York betting markets of 1884 to 1940 correctly predicted nearly every presidential election of their era \citep{rhode-strumpf2004}, despite being dominated by urban, male, financially active traders. Different from Kalshi's traders in almost every respect other than their lack of representativeness, the 1900 bettors priced the election accurately too.

These findings contribute the first large-scale empirical test of the claim that prediction market demographics bias prices, and with it the first direct confirmation of the marginal-trader hypothesis in a major election market. They also mark where the diversity argument for collective intelligence applies and where it does not. Prediction markets are not averaging mechanisms, so results on the wisdom of averaged crowds do not constrain them. Instead, our analysis suggests that financial incentives substitute for demographic breadth. The findings also bear on the policy debate over election contracts, where a central question is whether prices carry information or reflect the preferences of those who trade.

\subsection{Related literature}

The 2024 criticism borrows the authority of a familiar principle, popularized by \cite{surowiecki2004}, that diversity is required for crowds to forecast well. Surowiecki's diversity is diversity of perspectives, heuristics, and information, not demographic categories, as demographically similar individuals may hold heterogeneous information just as demographically diverse groups may think alike. Whether demographic diversity improves group performance is an empirical question with mixed answers elsewhere. Research teams that include both men and women produce more novel and more highly cited work \citep{yang2022}, while a mandated quota for women on corporate boards in Norway reduced firm value \citep{ahern-dittmar2012} and did little for women elsewhere in those firms \citep{bertrand2019}, and surveys of ethnic diversity find both costs and benefits for economic performance \citep{alesina-laferrara2005, putnam2007, hong-kacperczyk2010}. Demographic characteristics also affect individual preferences \citep{washington2008} and the speed at which groups with similar members learn \citep{golub-jackson2012}.

The no-trade theorem of \cite{milgrom-stokey1982} implies that a market of rational traders with common priors generates no trade at all, so some heterogeneity is necessary for a market to function, but the relevant heterogeneity need not be demographic. Models with noise traders \citep{kyle1985, black1986} show that unsophisticated participants improve price discovery by giving informed traders a counterparty. In these models, accuracy depends on whether financial incentives lead traders to act on their beliefs about outcomes and not on their preferences over outcomes. Demographic composition and price accuracy are, on this view, distinct objects. Evidence from macroeconomic forecasting points the same way. Partisan affiliation biases the forecasts of professional economists, whose Republican-affiliated members project higher growth, and forecast it less accurately, when a Republican holds the presidency \citep{kay2025}, whereas Kalshi's macroeconomic markets, in which every forecast carries a financial stake, predict inflation and unemployment releases at least as accurately as the Bloomberg survey consensus \citep{diercks2026}. 
\cite{manski2006} makes the aggregation point precise. A market price is not an average of participant beliefs, so results on the wisdom of averaged crowds \citep{hong-page2004, page2007} do not transfer to it. The composition question also has a direct antecedent. \cite{forsythe-nelson-neumann-wright1992} proposed the marginal-trader hypothesis to explain how the Iowa Electronic Markets produced accurate prices from a demographically narrow trader pool. Under the hypothesis, the average participant exhibits familiar judgment biases, and prices are set by a small group of active, unbiased traders at the margin \citep{oliven-rietz2004, berg2008}. 

Though the above theories originated in prior work, the hypothesis has not been tested directly on a major election market. The Iowa markets are experimental in scale and design. Each trader may invest at most \$500, the 1988 market studied by \cite{forsythe-nelson-neumann-wright1992} was open only to members of the University of Iowa, and the combined trading volume of the markets in these studies is less than 0.01\% of the volume on Kalshi's 2024 election markets. Whether the hypothesis holds in a large market open to the public has remained an open question. Concurrently, \cite{zitzewitz2026bias} studied prediction market returns from 1880 to 2025 and found no evidence that left-leaning or right-leaning candidates were overpriced, and no overpricing related to the demographic characteristics of candidates in recent elections. His data does not identify traders, so his analysis cannot identify why or how prices are unbiased, whether because traders are representative of voters or because unrepresentative traders do not trade on their politics. Here we study the traders directly, and therefore gain insight into the mechanism that drives prediction market accuracy.

 Appendix~\ref{app:model} formalizes our argument in a model of selected participation. When participation depends on voter type, the price distortion is the product of a selection gap and the degree to which traders' beliefs depend on their own type, each weighted by the dollars a trader commits. The market can be biased only if both factors are nonzero at the financially weighted margin.  We estimate these factors in Section~\ref{section:mechanism}.

The paper proceeds as follows. Section~\ref{section:data} describes the markets and the data. Section~\ref{section:accuracy} establishes the accuracy results. Section~\ref{section:demographics} documents the demographic composition of Kalshi's traders. Section~\ref{section:mechanism} presents the evidence on the incentives mechanism. Section~\ref{section:historical} places the findings in historical context, and Section~\ref{section:conclusion} concludes.

\section{Prediction Markets and Data}\label{section:data}\label{section:background}

Political prediction markets are financial exchanges on which participants trade contracts whose payoffs depend on election outcomes. The standard contract is a winner-take-all binary option paying \$1 if a specified candidate wins and \$0 otherwise, so that the price can be interpreted as the market's estimated probability of the outcome.\footnote{The mapping from price to belief depends on the distribution of beliefs and stakes \citep{manski2006, wolfers-zitzewitz2006}. Appendix~\ref{app:model} adopts the benchmark under which the price is the dollar-weighted mean belief.} Election betting has a long history in the United States. Organized markets operated openly on the curb exchanges of New York from the 1880s through the 1940s, handling large volumes and outperforming the straw polls of the era \citep{rhode-strumpf2004, rhode-strumpf2013}. The modern academic platforms, the Iowa Electronic Markets \citep{forsythe-nelson-neumann-wright1992, oliven-rietz2004} and PredictIt, operate at small scale under CFTC no-action letters yet have matched or beaten major polling aggregators \citep{wolfers-zitzewitz2004, crane2018, crane-vinson2020}. The scale of political prediction markets changed in the 2024 cycle. Polymarket, an offshore blockchain-based platform, transacted roughly \$3.6 billion on its presidential winner market \citep{gladwin2024}, and Kalshi, a CFTC-regulated exchange approved to list election contracts one month before the election, transacted more than \$250 million on its presidential winner market \citep{odonnell2024}. 

Increased scale of these markets brought increased scrutiny, as discussed in Section \ref{section:introduction}. The larger scale also allows better tests of the claims made about prediction markets, both favorable and unfavorable.  We address some of these questions here, with data compiled over a number of election cycles spanning decades of U.S.~political history. We describe the data sources below.

\paragraph{Kalshi.} Kalshi is a CFTC-regulated exchange. For this study Kalshi provided de-identified data on trading in its election markets between October 4, 2024, when it was first permitted to list election contracts, and November 5, 2024, election day. The data contains no names or other personally identifying information for any user. Trading in these markets totaled \$338 million in buy-side dollar volume. Buy-side volume grew from \$0.3 million on October 4 to \$114 million on election day, and the final five days of the campaign account for 63\% of the total (Figure~\ref{fig:volume}).

\begin{figure}[h!]
  \centering
  \includegraphics[width=0.75\textwidth]{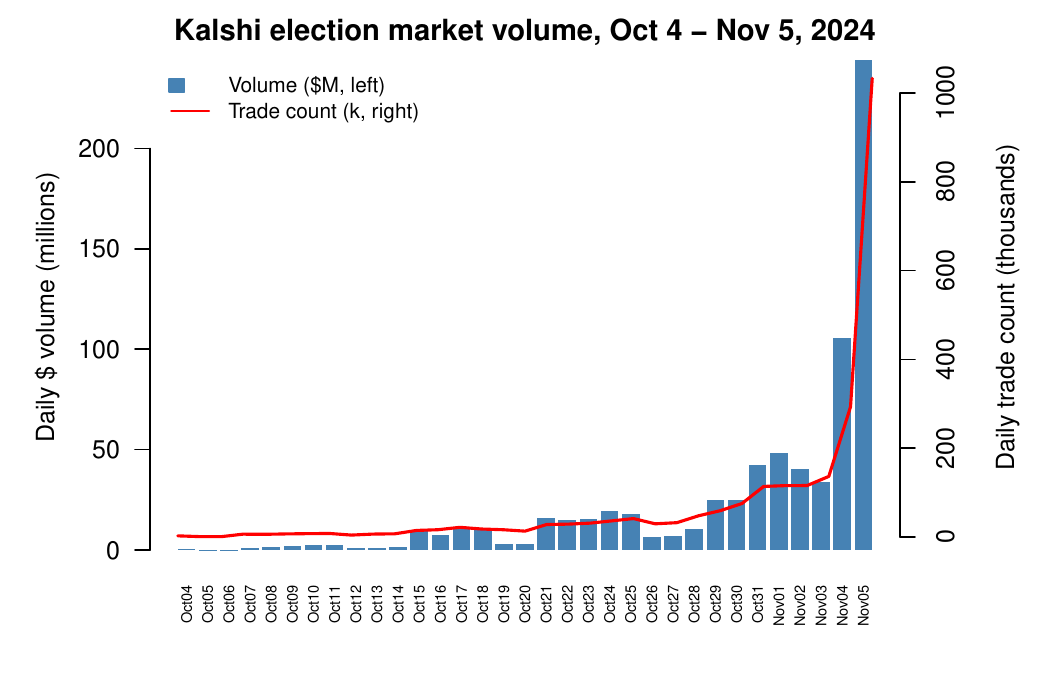}
  \caption{Daily dollar volume (bars, left axis) and trade count (line, right axis) on Kalshi's election markets, October 4 to November 5, 2024. Dollar volume here includes both buys and sells.}
  \label{fig:volume}
\end{figure}

\paragraph{Zip-code partisanship.} We measure the partisanship of each trader's home zip code in two ways. The primary measure is the two-party Republican share of the zip code's 2024 presidential vote. We construct it from the precinct-level results and precinct boundaries published by \cite{nyt2025precinct}. We overlay the precinct boundaries on a map of zip-code boundaries published by the mapping company TomTom in June 2024. Each precinct's votes are allocated to zip codes in proportion to the share of the precinct's area that lies in each zip code. The measure covers 99.2\% of users. The secondary measure is the two-party Republican share of the 2020 presidential vote in the zip code tabulation area (ZCTA), from \cite{mastrosavvas2024}. The two measures have a correlation of 0.98. Where the analysis groups zip codes, we classify them as Blue (under 40\% Republican), Lean D (40 to 50\%), Lean R (50 to 60\%), and Red (over 60\%). Registration and turnout are not observed at the zip-code level.

\paragraph{Campaign contributions.} From the DIME database \citep{bonica2023} we construct zip-code-level aggregates of individual campaign contributions in the 2024 cycle, counts and dollar amounts of donors to each party, including donors to the Trump and Harris/Biden presidential campaigns specifically. These provide a donor-based measure of zip-code partisanship, covering 99.6\% of users, and a measure of political engagement (donors per 1,000 residents). Because the Kalshi data is de-identified, individual traders cannot be matched to itemized contribution records. All contribution measures are zip-code aggregates.

\paragraph{Comparison forecasters.} We collected daily probability forecasts for the 2024 elections from three poll-based models (Nate Silver's Silver Bulletin, FiveThirtyEight, and The Economist) and daily prices from four market-based sources other than Kalshi. The market sources are Polymarket (presidential, state, Senate, governor, and House markets), PredictIt (the national race, Georgia and Wisconsin, and the Arizona and Ohio Senate races), the Betfair betting exchange, and a consensus of approximately two dozen UK sportsbooks compiled by Oddschecker.\footnote{Silver Bulletin forecasts are the daily national toplines published to subscribers, July 29 to November 5, 2024. FiveThirtyEight forecasts are from the model's published daily summary files (national and state presidential, and Senate win probabilities by race) and, for the House, from daily snapshots of its district-level win probabilities and House-control probability, October 9 to November 8. FiveThirtyEight published no governor forecast in 2024. The Economist forecasts are the national, state, and Senate time series published with its model. Polymarket prices are hourly token prices for the presidential winner market retrieved through Dune Analytics and, for the state, Senate, governor, and House markets (twenty-two district markets and House control), daily closing prices (7pm Eastern) of the Republican-wins contract from Polymarket's public order-book price history. Bookmaker and Betfair odds are daily Oddschecker snapshots of the national winner market, the 51 state winner markets, and seven Senate races. Fractional odds are converted to implied probabilities, normalized across the two major-party outcomes to remove the bookmaker margin, and the bookmaker consensus is the median across sportsbooks, excluding the Betfair and Matchbook exchanges. PredictIt prices are daily closing prices from the exchange's published price history, with the election-day open in place of the post-results close, normalized across the two major-party contracts to remove the fee-driven overround. The Economist files were received in January and July 2025, and the Silver Bulletin and FiveThirtyEight files in July and October 2025. FiveThirtyEight House snapshots and Oddschecker odds were collected daily through election day 2024. Polymarket prices were retrieved in October 2025 for the presidential winner market and in September 2026 for the other markets. The PredictIt price history was received in December 2025.}

\paragraph{PredictIt.} For external replication we use a separate proprietary user-level dataset from PredictIt's 2016 presidential market that includes registered party affiliation and verified federal donor status for each trader.

\paragraph{Historical markets.} Section~\ref{section:historical} draws on data from \cite{rhode-strumpf2004, rhode-strumpf2013} on the New York betting markets of 1884 to 1940, and on a new sample of 80 bettors in the 1900 election. The bettors were named in reports on the election betting in New York newspapers, including the \textit{New York Times}, the \textit{Wall Street Journal}, the \textit{New York World}, the \textit{New York Herald}, and the \textit{Brooklyn Eagle}. We matched the names to individual records of the 1900 Census, accessed through Ancestry.com.

\section{Forecasting Accuracy}\label{section:accuracy}

We begin with the accuracy of prediction markets relative to poll-based models in 2024. The accuracy analysis establishes the tension that the rest of the paper addresses. It is not our primary contribution.

A forecast can be judged on how close its probability was to the realized outcome (accuracy) or on whether its estimated probabilities align with empirically observed frequencies (calibration). The two measure different things. A poorly calibrated forecast can be accurate, and an inaccurate forecast can be calibrated. The Brier score used below is an accuracy measure in this sense, the squared error of a forecast probability against the realized outcome. Section~\ref{subsec:profit} applies a calibration test to individual traders.  

\subsection{Market prices and model forecasts over time}

Figure~\ref{fig:prices} plots the implied probability of a Trump victory from Polymarket, Kalshi, Silver Bulletin, FiveThirtyEight, and The Economist from July 1 through November 5, 2024. All five forecasts follow the same qualitative trajectory, but from mid-August onward the markets consistently priced Trump above the poll-based models. Figure~\ref{fig:relative} plots each forecast relative to Polymarket. With few exceptions every non-market forecast sat below the market price for the entire period.  Once Kalshi markets were listed in October, the Kalshi and Polymarket prices tracked each other closely.  The latter observation is what standard no-arbitrage reasoning predicts for markets that list similar contracts and whose prices are publicly observable. The tracking was not exact. Over the common window Polymarket's Trump price averaged 59.5\%, 2.4 percentage points above Kalshi's and one to two percentage points above the Betfair exchange and the bookmakers. During this period a single trader operating several accounts built a Trump position on Polymarket reported at more than \$30 million \citep{osipovich2024}. We do not attempt to estimate how much of the gap that position explains. 

\begin{figure}[h!]
  \centering
  \includegraphics[width=0.75\textwidth]{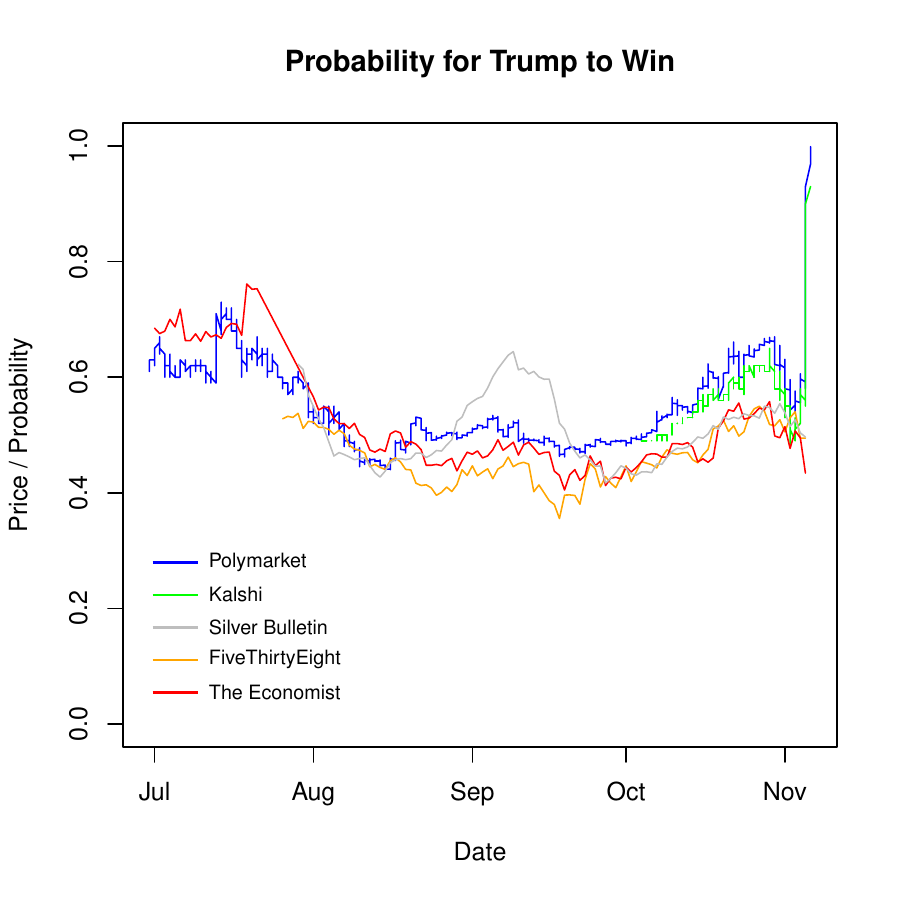}
  \caption{Probability of Trump winning the 2024 presidential election as implied by prediction market prices and poll-based forecasting models, July to November 2024.}
  \label{fig:prices}
\end{figure}

\begin{figure}[h!]
  \centering
  \includegraphics[width=0.75\textwidth]{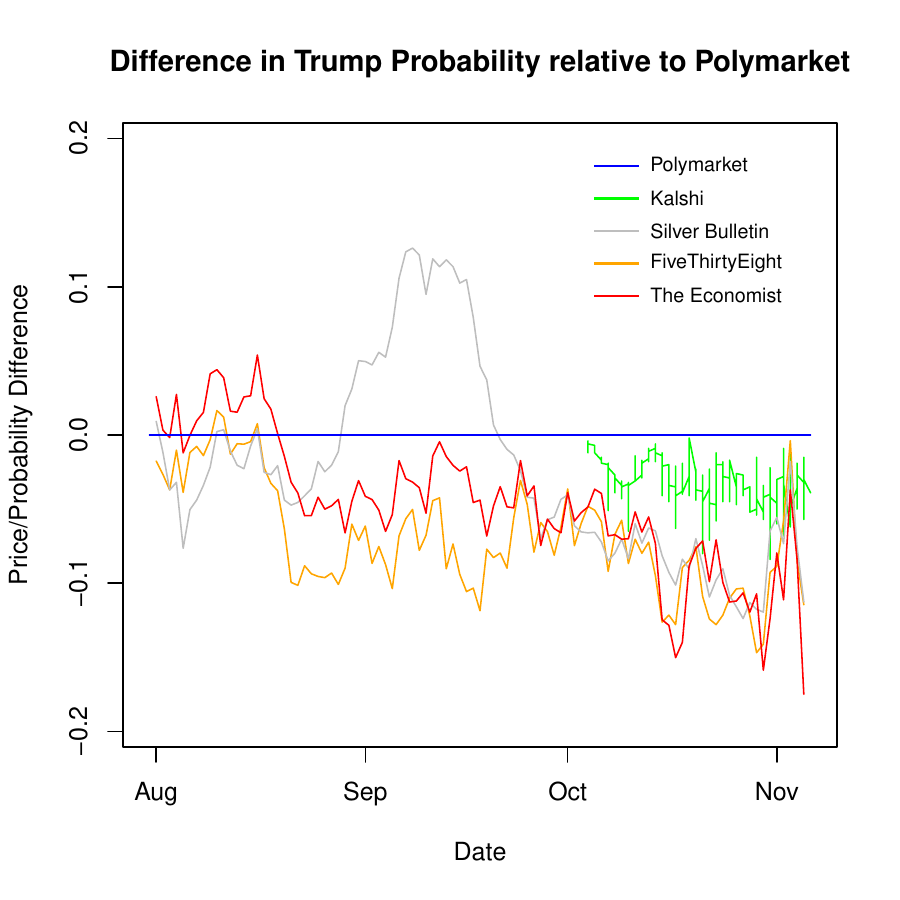}
  \caption{Difference in implied Trump win probability relative to Polymarket, August to November 2024. Values below zero indicate a lower Trump probability than Polymarket's.}
  \label{fig:relative}
\end{figure}

The poll-based models also moved in ways that no market price did. The clearest example is the hump in Silver Bulletin's forecast following the Democratic convention in late August, visible in the gray line of Figure~\ref{fig:relative}.  During this period, the Silver Bulletin model's Trump probability drifted upward for two weeks and then drifted back, a movement reflected in no market price and no competing model. Silver attributed the pattern to a model adjustment for expected post-convention polling bounces, which discounted Harris's polling averages by one to two percentage points in the weeks after the convention \citep{silver2024bounce}.  Because the polls did not move toward Harris after the convention, Silver's model interpreted flat polling results as a pro-Trump signal. Whatever the merits of a post-convention bounce adjustment in a typical election year, this episode suggests a defect in Silver's 2024 model and illustrates a structural difference between models and markets as forecasting mechanisms.  Model outputs depend on assumptions whose effects are often hidden and untested, while market prices are constantly tested and refined in real time by participants with money at stake.  In this case, even though Silver's forecast drifted in favor of Trump, the eventual winner, the forecast's temporary meander away from all other forecasts and then steady drift back toward those forecasts, while all other forecasts remained in the same cluster, indicates a model error that could be identified only in hindsight.  

\subsection{Brier score comparison}\label{subsec:brier}

We measure forecast quality with the Brier score, defined for a single forecast as $(p-o)^2$, where $p$ is the forecast probability and $o \in \{0,1\}$ is the realized outcome. Lower is better. We compute a daily Brier score for each forecaster and contest and average across days within the window in which the relevant Kalshi contracts traded.

This comparison extends the existing literature in two directions. First, we widen the set of forecasters beyond the single market-versus-model contrast that has been standard \citep{sethi2025}. For the national race we score eight forecasters, five market-based (Kalshi, Polymarket, PredictIt, the Betfair betting exchange, and a consensus of UK sportsbooks) and three poll-based (The Economist, FiveThirtyEight, and Silver Bulletin). Second, we widen the set of contests. Kalshi's election offerings included the presidential race in all 50 states, the District of Columbia, and the five district contests of Maine and Nebraska, as well as a number of congressional districts. We also pay separate attention to the seven battleground states (Arizona, Georgia, Michigan, Nevada, North Carolina, Pennsylvania, and Wisconsin), nineteen contested U.S. Senate seats, nine governor races, and the twenty-four U.S. House districts on which Kalshi opened contracts on October 31. The down-ballot contests are the more demanding test. They are decided by smaller margins, attract less forecasting effort, and were in several cases called only after election night.\footnote{We score realized outcomes against certified election results, not against the last in-sample contract price. For contests called after the close of our trade window, including the Michigan and Wisconsin Senate races, both won by Democrats, and House districts such as California's 13th, called on December 3, the final traded price does not coincide with the \$1/\$0 settlement, and using it would mislabel the outcome.}

\paragraph{National presidential race.} Table~\ref{tab:brier_summary} reports mean daily Brier scores across all forecasters for the 2024 U.S.~presidential race. Because Kalshi listed election contracts only on October 4, a like-for-like national comparison is restricted to the common window of October 7 through November 5. The window is the final month of the campaign.  During this period, the tracked forecasters separate into two groups with one forecaster between them. The four uncapped market-based forecasters score between 0.166 (Polymarket) and 0.185 (Kalshi), with the Betfair exchange (0.176) and the bookmaker consensus (0.179) between them. The three poll-based models score between 0.241 (Silver Bulletin) and 0.252 (FiveThirtyEight and The Economist). PredictIt, the academic exchange that caps each trader at \$850 per contract, scores 0.219, closer to the models than to the other markets. Its mean Trump probability over the window was 53\%, against 57 to 60\% on the uncapped markets and 50 to 51\% for the models. Every market beat every model, and the gap between the worst uncapped market and the best model (0.056) is larger than the spread within either group. The markets did not agree among themselves on the size of Trump's edge, Polymarket priced him highest and Kalshi lowest among the uncapped exchanges, but all of them sat above the models throughout the window (Figure~\ref{fig:brier_national}).\footnote{The market prices are not independent forecasts. Section~\ref{subsec:interpretation} discusses the dependence across venues. What their agreement shows is that the poll-based models, not the markets, were the outliers in 2024.}

\paragraph{Down-ballot contests.} The market advantage persists in the battleground states, the Senate races, and the governor races. Across the seven battleground states, the two exchanges score best (Polymarket 0.167, Kalshi 0.182), the Betfair exchange and the bookmakers follow (0.191 and 0.195), and the two poll-based models that published state forecasts trail (The Economist 0.213, FiveThirtyEight 0.230). Kalshi scores better than both models in every one of the seven states, and Polymarket records the best score of any forecaster in six of the seven (Table~\ref{tab:brier_states}). PredictIt listed only Georgia and Wisconsin among the battlegrounds and scores 0.193 on them against Kalshi's 0.182 on the same days. On its two Senate races, Arizona and Ohio, it scores 0.119 against Kalshi's 0.117. Across all 56 state and district contracts listed on Kalshi, the two exchanges again score best, Polymarket 0.040 and Kalshi 0.041, followed by The Economist 0.046, the bookmakers 0.047, and FiveThirtyEight 0.053. In the nineteen Senate races Kalshi listed, Kalshi and Polymarket are indistinguishable (0.090 each), FiveThirtyEight is close behind (0.094), and The Economist trails (0.108). On the seven Senate races the bookmakers priced, Kalshi scores 0.143 to their 0.156. In the nine governor races, which only the two exchanges forecast, Kalshi scores 0.100 to Polymarket's 0.108.

The House is the exception. Kalshi opened contracts on twenty-four House districts on October 31, six days before the election, and they drew \$236,000 in total, against \$47.5 million on the state presidential contracts and \$10.1 million on the Senate contracts. Over those six days FiveThirtyEight's district forecasts score 0.192 against Kalshi's 0.241, and Kalshi scores better in only four of the twenty-four districts. Polymarket, which listed sixteen of the districts, scores 0.229 against Kalshi's 0.223 on the same contest-days. On House control, a single contract that traded on 21 days of the window, the ordering is the same, FiveThirtyEight 0.233, Polymarket 0.261, Kalshi 0.264.\footnote{The House contracts are the thinnest markets in the analysis, and the six-day window is short. Kalshi's House-control contract drew \$424,000, more than all twenty-four district contracts together. The House sample is too small to separate a thin-market effect from a bad week, but the direction is the one the incentive account predicts, and the Senate contracts, with forty times the money, show no model advantage.} Figure~\ref{fig:brier_states} plots the daily battleground series and 
Figure~\ref{fig:brier_contests} the per-contest means for every state, Senate, governor, and House contest.
Kalshi's edge over the poll-based models holds in the national race, where the two main contracts drew \$168 million in buy-side volume, in the state contests (\$48 million), and in the Senate races (\$10 million). It is absent in the House districts (\$236,000) and on House control (\$424,000).

\paragraph{A profitability test.} Market-based scoring evaluates a forecaster by the performance of an investor who trades on a prediction market according to the forecaster's probabilities. \cite{crane2018} proposed it for the 2018 midterms and \cite{crane-vinson2020} developed the method more formally and applied it to the 2020 election. Crane and Vinson found that in 2020 a trader following FiveThirtyEight's forecasts would have profited against PredictIt's prices only on near-certain contracts. PredictIt's fees and position limits held the prices of those contracts below the maximum of \$1.00, so the profit reflects the exchange's rules, not better forecasting. For contests in which the result was in doubt, PredictIt outperformed the FiveThirtyEight forecast. \cite{sethi2025} adopted the same idea for 2024, arguing that the average daily Brier score can reward a model whose forecast drifts predictably toward the eventual outcome. 

We implement a simple version of these approaches here. On each contest-day, a trader who believes forecaster $f$'s probability $p_f$ on the Republican candidate to win and faces the Kalshi price $p_k$ on that outcome buys one contract on the side the forecaster favors, the Republican side at $p_k$ if $p_f > p_k$, the Democratic side at $1-p_k$ if $p_f < p_k$. Table~\ref{tab:brier_profit} reports the return per dollar staked, summed over contest-days within each group.
\begin{defbox}
No poll-based model would have profited by trading against Kalshi in any presidential contest group.
\end{defbox}
A trader who believed Silver Bulletin, FiveThirtyEight, or The Economist and traded the national race lost between 92 and 100 cents per dollar staked, and the losses in the battleground states run from 50 to 83 cents. In the Senate races and the full set of state contracts, where models and market largely agreed and the trader staked little, the returns sit near zero. The Economist loses 7 cents per dollar on the Senate races and FiveThirtyEight earns 5 cents, the same return as a trader who believed Polymarket. In the House districts a trader who believed FiveThirtyEight earned 58 cents per dollar against Kalshi, and 32 cents on House control. A trader who believed Polymarket's prices, by contrast, profited against Kalshi in the presidential contests (73 cents per dollar nationally, 49 cents in the battlegrounds) and lost in the governor races, where a trader who believed Kalshi and traded against Polymarket earned 23 cents per dollar. A trader who believed PredictIt's prices and traded against Kalshi lost 85 cents per dollar on the national race and 38 cents in Georgia and Wisconsin, and roughly broke even on the two Senate races. In the national race and the battleground states the ranking under the profitability test matches the Brier ranking, the exchanges first, the bookmakers next, the poll-based models last. The returns in Table~\ref{tab:brier_profit} are before fees and are negative for every poll-based model in every presidential group. On a capped, high-fee exchange in 2020 the leading model could not profit against the market. On an uncapped exchange in 2024 it loses outright. Where \cite{sethi2025} report that some statistical models beat the prediction market they study in less visible races, we find the same pattern in its narrowest form. No model profits against Kalshi in the presidential contests, which drew more than \$200 million, the Senate returns are within a few cents of zero, and the one model that clearly beats Kalshi does so in the House districts, where Kalshi's contracts were six days old and carried less than one dollar to every thousand traded on Kalshi's election contracts.\footnote{The rule is one of many possible. Sizing each stake in proportion to the disagreement $|p_f - p_k|$ gives the same signs and the same ordering in every cell. As with the Brier scores, contest-days are correlated.}

\begin{table}[h!]
\centering
\caption{Mean daily Brier score by forecaster and contest group, 2024. Lower is better. Bold marks the best score in each column. National scores use the common window (October 7 to November 5) over which all forecasters were active. For each down-ballot contest the window is the set of days on which the Kalshi contract traded within that period. The House contracts traded from October 31. A dash indicates that no forecast for that contest group is included.}
\label{tab:brier_summary}
\footnotesize
\setlength{\tabcolsep}{4pt}
\begin{tabular}{lcccccc}
\toprule
Forecaster & National & Battleground & All states & Senate & Governor & House \\
\midrule
Kalshi           & 0.185 & 0.182 & 0.041 & \textbf{0.090} & \textbf{0.100} & 0.241 \\
Polymarket       & \textbf{0.166} & \textbf{0.167} & \textbf{0.040} & \textbf{0.090} & 0.108 & 0.229$^{d}$ \\
PredictIt        & 0.219 & --$^{c}$ & -- & --$^{c}$ & -- & -- \\
Betfair Exchange & 0.176 & 0.191 & 0.105$^{a}$ & -- & -- & -- \\
Bookmakers       & 0.179 & 0.195 & 0.047 & 0.156$^{b}$ & -- & -- \\
Silver Bulletin  & 0.241 & --    & --    & --    & -- & -- \\
FiveThirtyEight  & 0.252 & 0.230 & 0.053 & 0.094 & -- & \textbf{0.192} \\
The Economist    & 0.252 & 0.213 & 0.046 & 0.108 & -- & -- \\
\midrule
Contests                    & 1  & 7   & 56    & 19  & 9  & 24 \\
Kalshi forecaster-days      & 30 & 210 & 1,047 & 370 & 67 & 139 \\
\bottomrule
\end{tabular}
\begin{minipage}{0.92\textwidth}
\footnotesize $^{a}$ The Betfair exchange listed 25 of the 56 state contracts, weighted toward competitive states. Kalshi scores 0.101 on the same contest-days. $^{b}$ The bookmakers priced seven Senate races. Kalshi scores 0.143 on the same contest-days. $^{c}$ PredictIt's 2024 slate included two battleground states, Georgia and Wisconsin (0.193 against Kalshi's 0.182 on the same contest-days), and two Senate races, Arizona and Ohio (0.119 against 0.117). $^{d}$ Polymarket listed sixteen of the twenty-four House districts. Kalshi scores 0.223 on the same contest-days. On House control, a single contract traded on 21 days of the window, FiveThirtyEight scores 0.233, Polymarket 0.261, and Kalshi 0.264.
\end{minipage}
\end{table}

\begin{table}[h!]
\centering
\caption{Mean daily Brier score by battleground state, October 7 to November 5, 2024. Bold marks the best score in each state. The two exchanges record the best score in every state, and Kalshi scores better than both poll-based models in every state.}
\label{tab:brier_states}
\begin{tabular}{lccccccc}
\toprule
Forecaster & AZ & GA & MI & NC & NV & PA & WI \\
\midrule
Kalshi           & 0.094 & 0.118 & 0.281 & 0.124 & 0.207 & 0.206 & 0.245 \\
Polymarket       & 0.084 & \textbf{0.108} & \textbf{0.250} & \textbf{0.118} & \textbf{0.197} & \textbf{0.181} & \textbf{0.229} \\
PredictIt        & --    & 0.125 & --    & --    & --    & --    & 0.261 \\
Betfair Exchange & \textbf{0.084} & 0.120 & 0.296 & 0.131 & 0.227 & 0.217 & 0.265 \\
Bookmakers       & 0.088 & 0.125 & 0.299 & 0.137 & 0.224 & 0.219 & 0.271 \\
The Economist    & 0.121 & 0.147 & 0.334 & 0.132 & 0.273 & 0.226 & 0.262 \\
FiveThirtyEight  & 0.143 & 0.160 & 0.313 & 0.177 & 0.266 & 0.251 & 0.299 \\
\bottomrule
\end{tabular}
\end{table}

\begin{table}[h!]
\centering
\caption{Profitability test: return per dollar staked by a trader who believes each forecaster's probability and buys one contract per contest-day against the Kalshi price, on the side the forecaster favors. Positive values mean the forecaster would have beaten Kalshi. The final row reverses the roles, with Polymarket as the market and Kalshi as the believed forecaster. A dash indicates no forecast for that contest group.}
\label{tab:brier_profit}
\footnotesize
\setlength{\tabcolsep}{4pt}
\begin{tabular}{lcccccc}
\toprule
Believed forecaster & National & Battleground & All states & Senate & Governor & House \\
\midrule
\multicolumn{7}{l}{\textit{Trading against Kalshi}} \\
Polymarket       & $+0.73$ & $+0.49$ & $+0.11$ & $+0.05$ & $-0.10$ & $+0.00^{b}$ \\
PredictIt        & $-0.85$ & $-0.38^{a}$ & -- & $+0.02^{a}$ & -- & -- \\
Betfair Exchange & $+0.51$ & $-0.36$ & $-0.13$ & --      & -- & -- \\
Bookmakers       & $+0.33$ & $-0.45$ & $-0.13$ & $-0.20$ & -- & -- \\
Silver Bulletin  & $-0.92$ & --      & --      & --      & -- & -- \\
FiveThirtyEight  & $-0.92$ & $-0.83$ & $-0.08$ & $+0.05$ & -- & $+0.58$ \\
The Economist    & $-1.00$ & $-0.50$ & $-0.01$ & $-0.07$ & -- & -- \\
\midrule
\multicolumn{7}{l}{\textit{Trading against Polymarket}} \\
Kalshi           & $-0.92$ & $-0.54$ & $-0.15$ & $-0.01$ & $+0.23$ & $+0.19^{b}$ \\
\bottomrule
\end{tabular}
\begin{minipage}{0.92\textwidth}
\footnotesize $^{a}$ PredictIt's two battleground states (Georgia and Wisconsin) and two Senate races (Arizona and Ohio) only. $^{b}$ The sixteen House districts listed on both exchanges. On House control, a FiveThirtyEight believer earned $+0.32$ per dollar against Kalshi and $+0.40$ against Polymarket.
\end{minipage}
\end{table}

\begin{figure}[h!]
  \centering
  \includegraphics[width=0.85\textwidth]{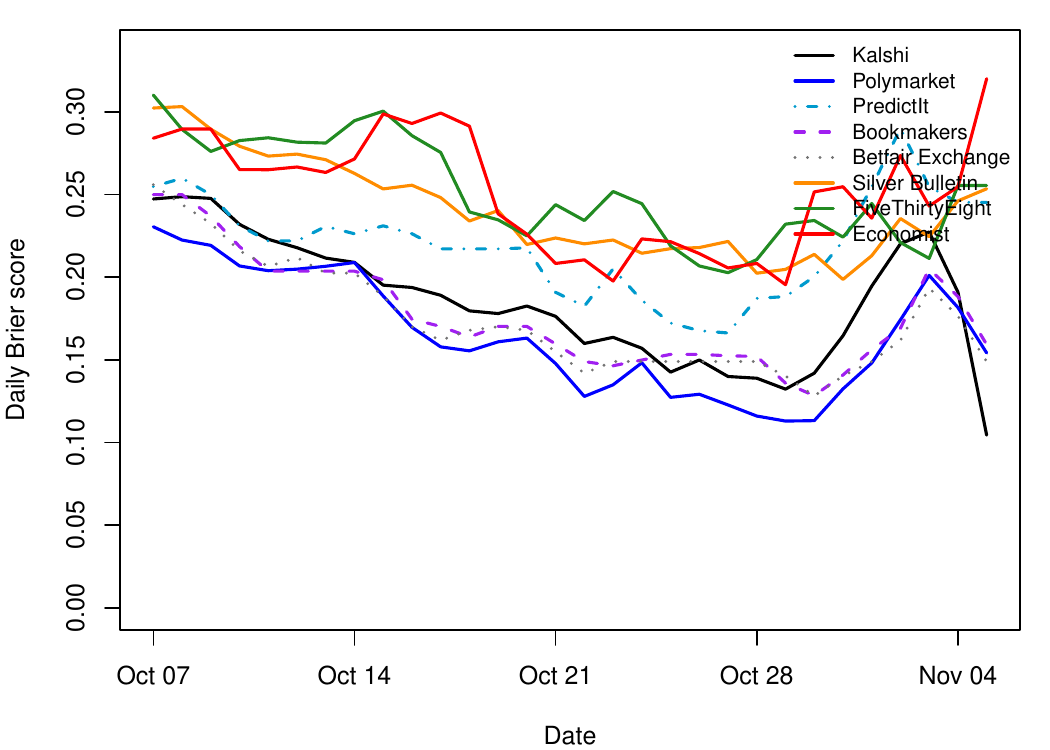}
  \caption{Daily Brier scores for the national presidential race, by forecaster, over Kalshi's listing window. Lower values indicate better accuracy.}
  \label{fig:brier_national}
\end{figure}

\begin{figure}[h!]
  \centering
  \includegraphics[width=0.95\textwidth]{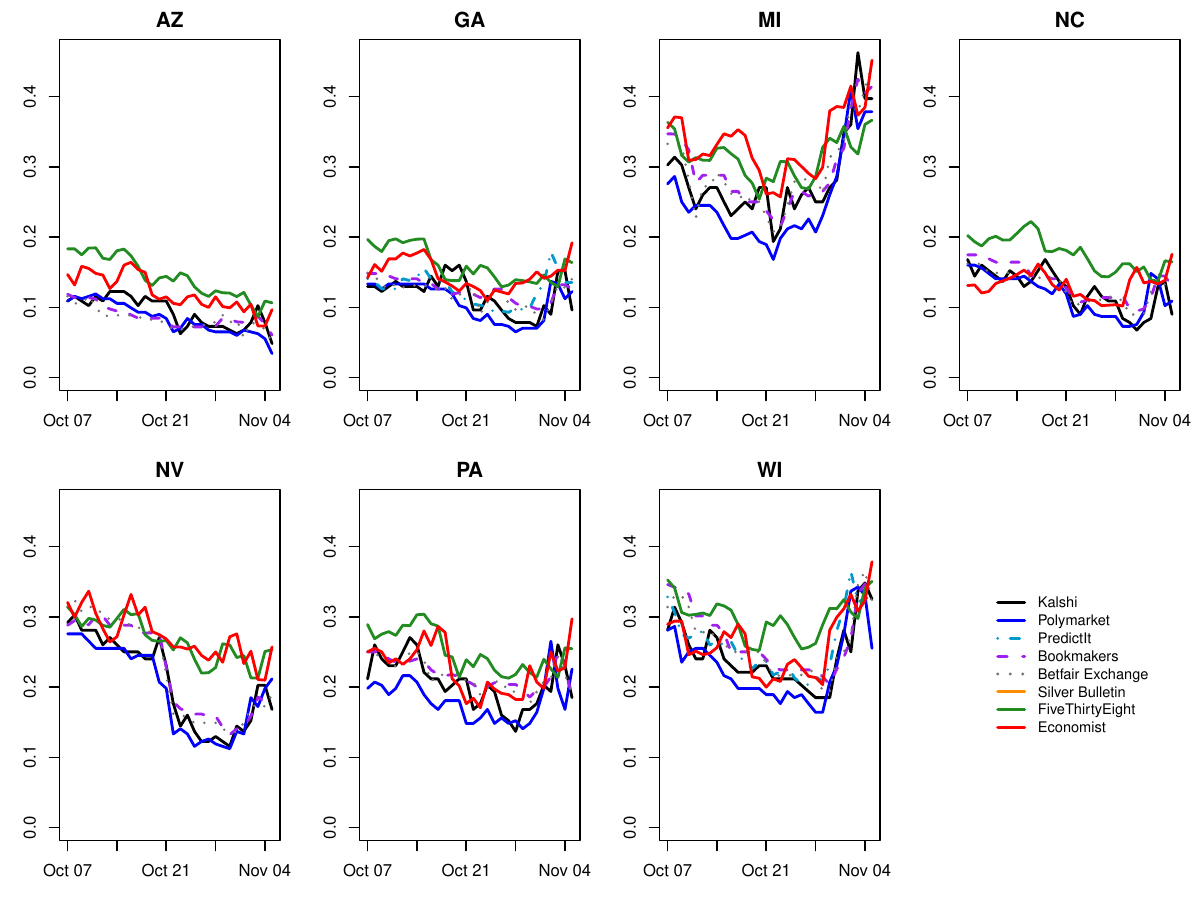}
  \caption{Daily Brier scores for the seven battleground-state presidential races, by forecaster, October 7 to November 5, 2024.}
  \label{fig:brier_states}
\end{figure}

\begin{figure}[h!]
  \centering
  \includegraphics[width=0.98\textwidth]{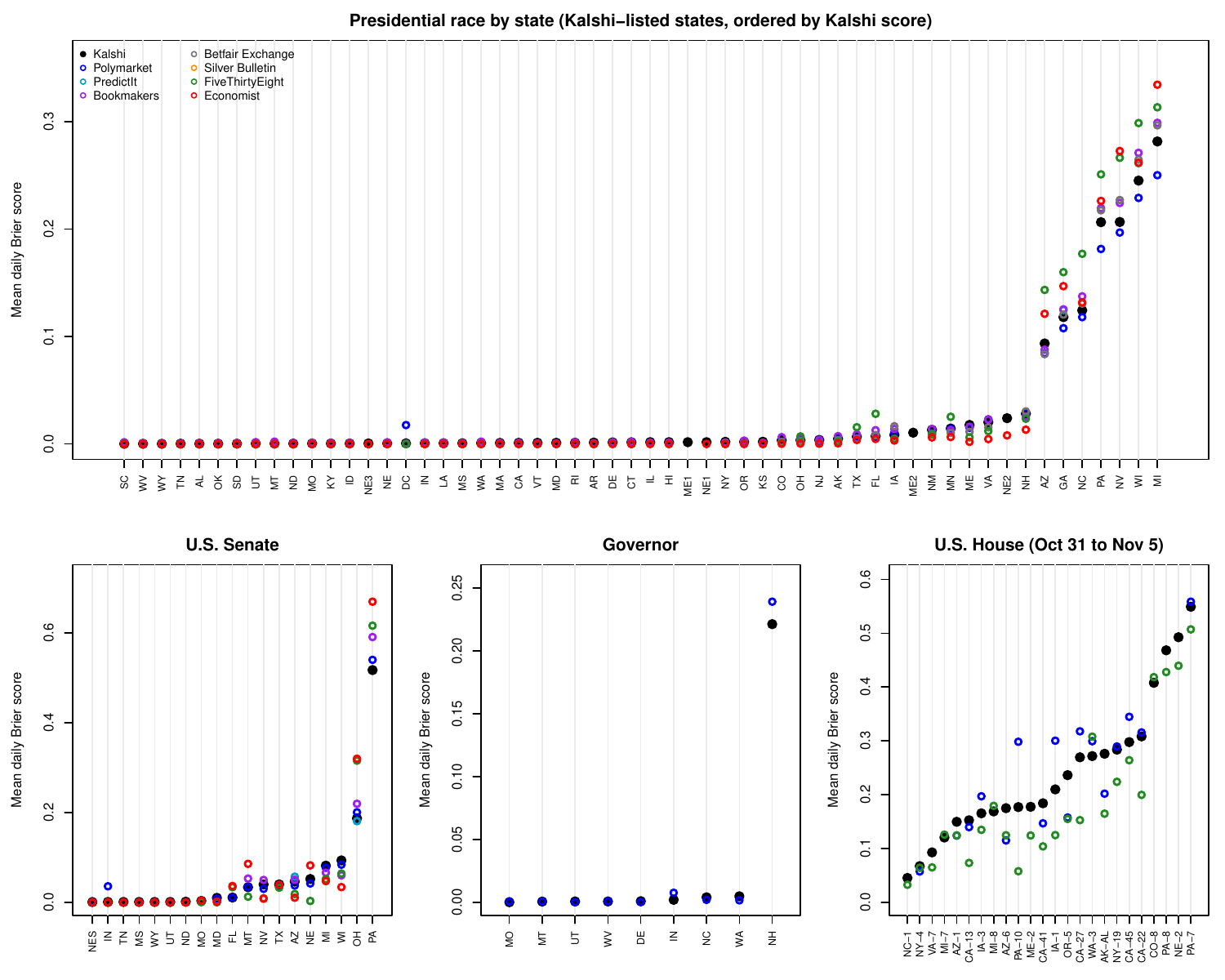}
  \caption{Mean daily Brier score by contest and forecaster, October 7 to November 5, 2024. Top: the presidential race in the 56 states and districts Kalshi listed, ordered by Kalshi's score. Bottom: nineteen Senate races, nine governor races, and the twenty-four House districts Kalshi listed from October 31. Filled black points are Kalshi. Open points are the other forecasters. Most contests were near-certain and every forecaster scores near zero on them. The forecasters separate only in the competitive contests at the right of each panel.}
  \label{fig:brier_contests}
\end{figure}

The pattern across contests is consistent.
\begin{defbox}
Every market-based forecaster was more accurate than every poll-based model in the national race, and Kalshi was more accurate than every poll-based model in the state and Senate races.
\end{defbox}
Markets on House of Representatives districts, which only ran for six days on Kalshi, are an exception. In those markets, low volume reduced the economic incentives, and accordingly the accuracy of those markets suffered. The puzzle that the rest of the paper addresses is that an exchange which outperformed every poll-based model in the national, state, and Senate races was, as we show next, composed of traders who look nothing like the voting population.

\section{Kalshi's Traders Are Not Representative}\label{section:demographics}

The pro-Trump lean of the markets relative to the models admits two explanations, superior forecasting or the political preferences of an unrepresentative trader base. Deciding between them begins with establishing how unrepresentative the trader base was. This section documents the composition of Kalshi's traders against the citizen voting-age population (CVAP) from the American Community Survey.

\paragraph{Sex.} Approximately 53\% of 2024 voters were female according to national exit polls.\footnote{National exit poll conducted by Edison Research for the National Election Pool, November 2024, as published by NBC News.} On Kalshi, 85.8\% of election-market traders were male, and male traders accounted for 96.9\% of dollar volume. Table~\ref{tab:sex_age} reports the joint distribution of traders by sex and age.

\begin{table}[ht]
\centering
\caption{Kalshi traders by age and sex (\% of traders with imputed sex).}
\label{tab:sex_age}
\begin{tabular}{lrr}
\toprule
Age & Male (\%) & Female (\%) \\
\midrule
18--29 & 44.88 & 5.74 \\
30--34 & 13.07 & 2.23 \\
35--39 & 10.26 & 1.83 \\
40--49 & 11.10 & 2.31 \\
50--59 &  4.50 & 1.38 \\
60+   &  2.02 & 0.65 \\
\midrule
Total & 85.83 & 14.14 \\
\bottomrule
\end{tabular}
\end{table}

\paragraph{Age.} Table~\ref{tab:age_dist} reports the age distribution of traders alongside the CVAP. Over half of Kalshi traders (51.5\%) were under 30, against 19.6\% of the CVAP. Only 7.9\% of traders were 50 or older, against 55.7\% of the CVAP aged 45 and up. Weighted by dollar volume the skew is less extreme but still pronounced, with traders under 30 accounting for 42.1\% of volume.

\begin{table}[ht]
\centering
\caption{Age distribution of Kalshi traders and of the citizen voting-age population (CVAP, Current Population Survey, November 2024). The CVAP tabulation uses different age groups, shown in the last two columns.}
\label{tab:age_dist}
\begin{tabular}{lrrcr}
\toprule
Age & \% of Kalshi traders & \% of Kalshi volume & CVAP age & CVAP share (\%) \\
\midrule
18--29 & 51.5 & 42.1 & 18--29 & 19.6 \\
30--34 & 15.5 & 23.6 & 30--44 & 24.7 \\
35--39 & 12.0 & 10.9 & 45--64 & 31.1 \\
40--49 & 13.1 & 16.6 & 65 and over & 24.6 \\
50--59 &  5.5 &  5.8 & & \\
60 and over &  2.4 &  1.1 & & \\
\bottomrule
\end{tabular}
\end{table}

\begin{figure}[h!]
  \centering
  \includegraphics[width=0.8\textwidth]{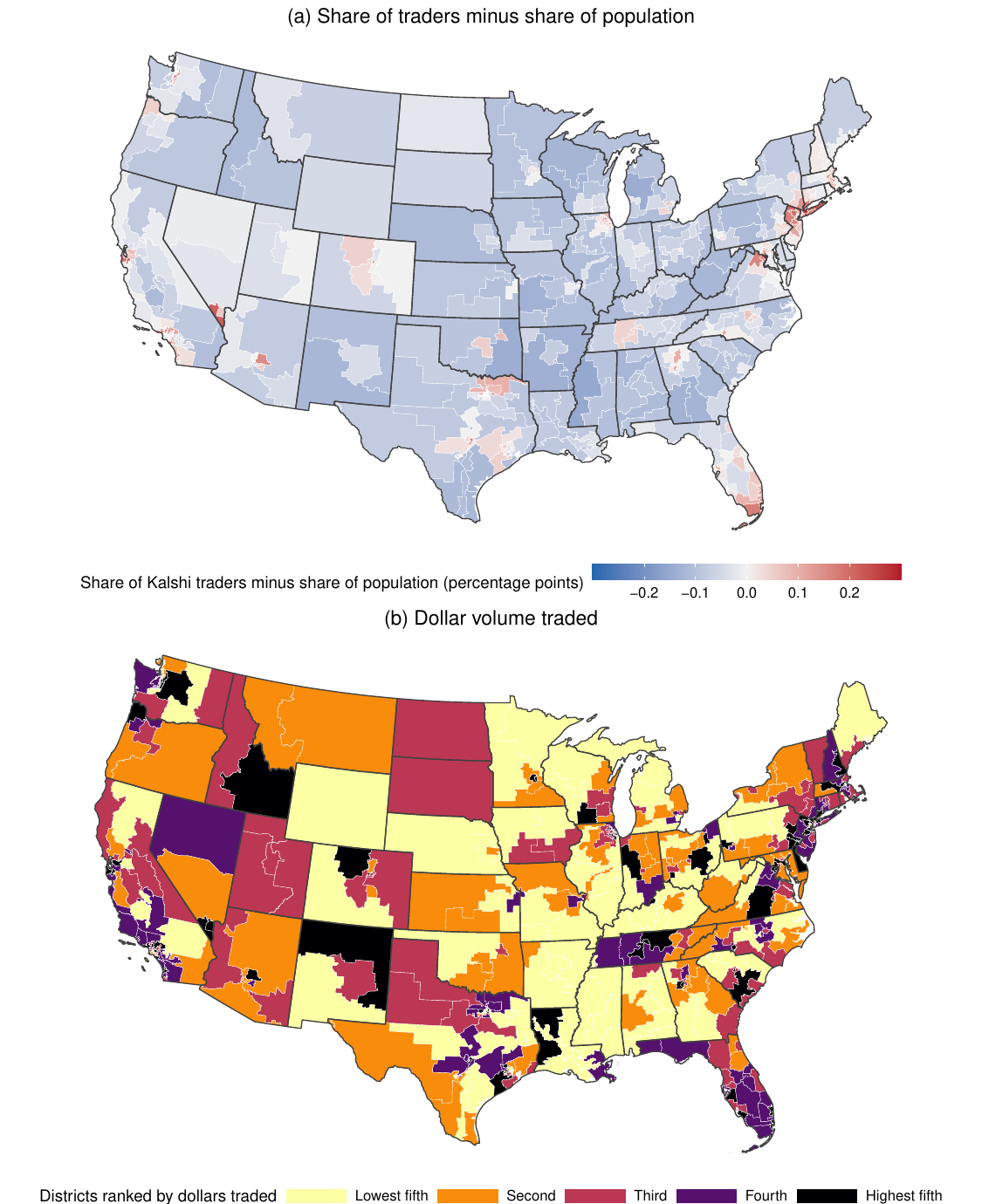}
  \caption{Kalshi election-market participation by congressional district, October 4 to November 5, 2024. Zip-code totals are summed to the districts of the 119th Congress, the districts contested in November 2024, with each zip code assigned to the district that contains its Census internal point, and district populations are summed from the same zip codes. The map shows the 432 districts of the contiguous states and the District of Columbia. Panel (a) shows each district's share of all Kalshi traders minus its share of the population, in percentage points, with the color scale clipped at 0.3 percentage points in either direction. Panel (b) ranks districts by dollar volume, buys and sells combined, in fifths. Volume is concentrated in a small number of urban districts that are too small to see at this scale, and the text reports the concentration.}
  \label{fig:zip_map}
\end{figure}

\paragraph{Geography.} Figure~\ref{fig:zip_map} maps Kalshi election-market participation by congressional district. Traders are disproportionately located in cities relative to the general population, and trade volume is even more concentrated. Panel (a) of Figure~\ref{fig:zip_map} puts a number on the first claim. A district's share of Kalshi's traders differs from its share of the population by between $-0.15$ and $+0.96$ percentage points, and moving 17.5\% of traders across districts would make the two shares match everywhere. Traders are over-represented in the urban districts of New York City, San Francisco, Chicago, and northern Virginia, and under-represented across the rural South and Midwest. The top 1\% of zip codes by volume account for 9.6\% of traders and 76.7\% of volume, and the top decile for 43\% of traders and 92\% of volume. Participation is also geographically unrepresentative. Among zip codes with at least 500 residents and 50 two-party votes, 31\% have no Kalshi trader. These zip codes hold 5\% of the population and cast 5\% of the votes, and they voted 68\% Republican. Among zip codes with traders, the population-weighted Republican vote share falls monotonically from 55\% in the bottom quintile of traders per 1,000 residents to 40\% in the top quintile (Table~\ref{tab:zip_intensity}). High-participation zip codes are also the politically engaged ones. Campaign donors per 1,000 residents more than double across the same quintiles, and in a regression of log traders per 1,000 residents on zip-code characteristics the elasticity with respect to donor density is 0.40, with the Republican vote share entering negatively ($-0.45$). Trader-level volume shows the same tilt. Users from Blue areas traded about ten times as much per user as users from Red zip codes, a difference that reflects the concentration of heavy traders in Democrat-dominated locations. Table~\ref{tab:zip_intensity} shows that the highest-participation quintile of zip codes nonetheless committed the largest share of its volume to Republican positions of any group, 62\%, a first indication that where traders come from does not dictate how they trade.

\begin{table}[h!]
\centering
\caption{Zip-code characteristics by quintile of Kalshi participation. Zip codes with at least 500 residents and 50 two-party votes. Quintiles of traders per 1,000 residents among zip codes with at least one trader, with the 31\% of zip codes that have no trader shown separately. Vote share is the population-weighted 2024 two-party Republican share. Donors and contributions are 2024-cycle individual contributors and dollars from DIME. The last column is the share of the quintile's classified election-market dollars committed long Republican.}
\label{tab:zip_intensity}
\begin{tabular}{lcccccc}
\toprule
 & Traders & Kalshi \$ & Vote share & Donors & Contributions & Long \\
Zip codes & per 1,000 & per capita & Republican & per 1,000 & \$ per capita & Republican \\
\midrule
No traders                  & 0    & 0     & 68\% & 29 & 22  & --   \\
Quintile 1 (lowest)         & 0.20 & 0.69  & 55\% & 27 & 18  & 53\% \\
Quintile 2                  & 0.37 & 1.30  & 52\% & 32 & 34  & 58\% \\
Quintile 3                  & 0.52 & 2.26  & 50\% & 39 & 57  & 59\% \\
Quintile 4                  & 0.72 & 5.43  & 46\% & 46 & 50  & 55\% \\
Quintile 5 (highest)        & 1.36 & 42.85 & 40\% & 63 & 178 & 62\% \\
\bottomrule
\end{tabular}
\end{table}

The geographic skew can be summarized in one number, the difference between the average Republican vote share of traders' home zip codes and the average for the population. We call this difference the selection gap on the geographic dimension. It is one of the two quantities that determine how far an unrepresentative trader pool can move the price in the model of Appendix~\ref{app:model}, which Section~\ref{subsec:framework} summarizes. Weighting zip codes by residents, the average is 49.3\% Republican. Weighting by traders it is 44.8\%, a selection gap of 4.5 percentage points. Weighting by volume it is 27.7\%, a selection gap of 21.6 percentage points. For traders with more than 100 trades, who account for 93\% of volume, the volume-weighted selection gap is 22.6 percentage points.
\begin{defbox}
Kalshi's 2024 election volume came from communities that voted nearly three to one against the candidate the market favored.
\end{defbox} 

On sex and age the gap points the other way, toward the winner. Male traders supplied 97\% of volume and half of traders were under 30. Men voted 55 to 43 for Trump, so an expressive market with this composition would have tilted toward Trump. Section~\ref{subsec:direction} shows that among traders with ten or fewer trades male traders put a smaller share of their volume on the Republican side than female traders did, and that among more active traders the two sexes do not differ, so this channel is closed as well.

These results confirm the premise of the critique. 
\begin{defbox}
In 2024, prediction market traders were substantially younger and more male than the voting population, and their trading volume came disproportionately from Democratic-voting, donor-dense communities.
\end{defbox}
The typical Kalshi trader is unrepresentative of the voting public in several respects, and the market's prices show no obvious signs of bias in the direction of traders' demographics. We do not have comparable data for Polymarket, though as an offshore exchange that is technically closed to U.S. customers its trader base is likely at least as skewed, consisting of a mix of non-U.S.~users (who are not voters) and U.S.~users who are avid crypto users. \cite{crane2026offshore} estimates that U.S. users accounted for about 30\% of Polymarket's trading volume in the twelve months ending April 2026, with a range of 19\% to 48\%, suggesting that most of Polymarket's volume comes from traders who could not vote in the 2024 election. Setting politics aside, Polymarket's trader base is on this dimension the least representative of the electorate of any market we consider, yet its prices for the 2024 election were consistently the most accurate based on the metrics studied here.

Paired with the accuracy results of Section~\ref{section:accuracy}, the composition results establish the paper's central tension. Demographic representativeness is not necessary for forecasting accuracy. The trader pool that beat every poll-based model in the national, state, and Senate races looks nothing like the electorate whose behavior it was forecasting. The remainder of the paper explains why.

\section{Incentives, Not Demographics}\label{section:mechanism}

Our proposed resolution distinguishes between who the traders are and how they behave. A prediction market trader who prefers Harris but believes Trump will win has a financial incentive to buy Trump contracts. A poll respondent asked for their preference faces no such distinction. If incentives govern behavior, the dollar-weighted actions of traders should converge on the market's belief about the outcome regardless of the traders' politics. If preferences govern behavior, trading is expressive and traders' positions should track their politics. Traders from Republican communities should buy Trump (on average), traders from Democratic communities should buy Harris (on average), and the price should aggregate the trader population's political composition.

We test these two accounts against each other in this section.  In each test, we split traders by the number of trades they placed, since low-frequency, low-volume traders are the cohort that is least affected by the market's profit incentives, and high-frequency, high-volume traders are the most. The two accounts make opposite predictions.
\begin{enumerate}
\item If preferences govern trading, the partisan pattern should be about the same in every cohort. Republican-aligned traders bet on Trump and other Republican candidates disproportionately more than Democrat-aligned traders, regardless of their trading frequency.
\item If incentives govern trading, the partisan pattern should be concentrated among casual traders and absent among the heavy traders, who are most directly affected by the incentives of the market.
\end{enumerate}

The evidence supports the incentives account.
\begin{defbox}
Traders with more than 100 trades account for 93\% of volume in Kalshi's 2024 election markets. Weighted by dollar volume, the direction of trades from this group is uncorrelated with the partisanship of their home zip code, sex, or age.
\end{defbox}
Under the framework of Section~\ref{subsec:framework}, this is the condition for an undistorted price. The price is set by the volume of traders whose positions do not follow their politics, so the composition of the trader pool, however unrepresentative, does not reach the price. The partisan pattern exists, but only among the casual traders who account for the remaining 7\% of volume. Among traders with ten or fewer trades, a 10-percentage-point increase in the home zip code's Republican vote share is associated with a 5-percentage-point increase in the share of volume placed on Republican positions, and among traders with 11 to 100 trades with a 3-percentage-point increase. 

\subsection{Framework}\label{subsec:framework}

Appendix~\ref{app:model} sets out a model of selected participation that organizes the tests. We summarize its objects here. Each voter has a latent type $v$ that affects the vote and may also affect participation in the market, so that the distribution of types among traders, $F_T$, differs from the distribution among voters, $F_V$. The \emph{selection gap} $\Delta_T = E[v \mid \text{trader}] - E[v]$ measures the difference. The market price is a weighted average of trader beliefs $q_i$ about the outcome, with weights $\lambda_i$ proportional to the dollars each trader commits, so the distribution that matters for the price is not $F_T$ but the dollar-weighted distribution $F_{\$}$, and the financially weighted selection gap $\Delta_{\$}$ is the relevant one. The price error is the financially weighted average belief error,
\[
B \;=\; \frac{\sum_i \lambda_i \,(q_i - q^*)}{\sum_i \lambda_i},
\]
where $q^*$ is the true conditional probability of the outcome. Representativeness enters nowhere in this expression. The price is unbiased if and only if $B = 0$, whatever $F_T$ looks like (Theorem~\ref{thm:prob}).

The appendix considers four accounts of how traders form the beliefs $q_i$, each a separate behavioral model.
\begin{enumerate}
\item Forecasting from own type. A trader forecasts from the trader's own type, and the price is distorted by the covariance between participation and the belief that type induces (Proposition~\ref{prop:owntype}).
\item Correct beliefs. Every trader holds the correct belief, and the price is correct whatever the composition of the trader pool (Proposition~\ref{prop:correct}).
\item Inference from the trader population. Traders take the trader population for the electorate, and the price is distorted to the extent that the two populations imply different vote shares (Proposition~\ref{prop:pool}).
\item Own-type projection. A trader perceives the electorate as shifted toward the trader's own type by a factor $\rho$ (Proposition~\ref{prop:projection}).
\end{enumerate}
Expressive trading is observationally a case of the first account, since the side taken depends on the trader's type alone.

The diversity critique is a claim about how $B$ is generated. If traders forecast from their own type, or perceive the electorate through the lens of their own type, a trader's belief depends on $v_i$ with some strength $\rho$, and the price distortion is, to first order, the product $\rho\,\Delta_{\$}$ scaled by the sensitivity of the win probability to the perceived electorate (Proposition~\ref{prop:projection}). The critique observes that $\Delta_T$ is large and concludes that $B$ is large. The conclusion requires two further things the critique does not establish, that $\rho$ is nonzero for the traders whose volume carries weight, and that $\Delta_{\$}$ has the sign the critique asserts. Financial incentives bear on the first. A trader who buys the preferred outcome regardless of the price has $\rho > 0$ and loses money on average. Repeat play is the process that punishes such a trader and rewards the trader with $\rho = 0$, so the incentives account predicts that $\rho$ falls with trade intensity and is zero where volume concentrates.

The sections that follow estimate each factor. Section~\ref{section:demographics} has shown that $\Delta_T$ is large on every demographic dimension, and that $\Delta_{\$}$ on the geographic dimension has the opposite sign from the one the critique assumes. Sections~\ref{subsec:geo} and \ref{subsec:lpm} estimate $\rho$ by trade intensity, from the partisan gap in trade direction under fractional attribution and from within-cohort regression slopes. Section~\ref{subsec:profit} reports an observable counterpart of $B$ for each cohort. The calibration score compares the prices a cohort paid with how often its positions won. The incentives account predicts that $\rho$ and $B$ both go to zero in the cohort that sets prices. The preference account predicts that neither does.

\subsection{Trader cohorts}\label{subsec:cohort}

We partition traders into five cohorts by trade count, exactly one trade (``=1''), 2--10, 11--100, 101--5,000, and more than 5,000. A single trade looks like expression, a vote cast with money. A user placing thousands of trades is monitoring the market continuously or trading by algorithm. Table~\ref{tab:cohort_summary} summarizes the cohorts.

\begin{table}[ht]
\centering
\caption{Kalshi trader cohorts by trade count: composition, behavior, and outcomes.}
\label{tab:cohort_summary}
\small
\begin{tabular}{lrrrrr}
\toprule
Cohort & \% of users & \% of volume & \% male & \% under 30 & ROI (\%) \\
\midrule
=1            & 13.0 &  0.1 & 79.4 & 47.5 & -55.9 \\
2--10         & 57.4 &  1.9 & 85.3 & 53.2 & -9.6  \\
11--100       & 24.3 &  4.7 & 89.2 & 51.4 & -9.5  \\
101--5{,}000  &  5.1 & 14.6 & 92.1 & 43.8 & -1.6  \\
\textgreater 5{,}000 &  0.2 & 78.7 & 95.5 & 32.9 & +1.5  \\
\bottomrule
\end{tabular}
\\[0.25em]
\footnotesize ROI is total PnL net of fees divided by total dollar volume traded by the cohort. Volume shares are of gross dollar volume across all election markets.
\end{table}

Activity is highly concentrated. Traders with more than 5,000 trades are 0.2\% of users and 78.7\% of dollar volume. ROI also rises monotonically with trade count, from the one-trade cohort to the heavy cohort. This is the pattern expected in a market that mixes expressive traders with traders who trade for profit.

The concentration of volume also explains how a market can price accurately despite widespread bias among its participants. Every trade moves the price toward the side traded, and large trades move it more. Large trades come overwhelmingly from the most active cohort, which also had the highest returns. The trades that move the price most therefore came from the traders whose positions proved most accurate on average. The rest of this section shows that this price-setting cohort is the one whose behavior shows no partisan pattern.

\subsection{Direction of trades across cohorts}\label{subsec:direction}

The cohorts differ most in which side they took. We define a position as long Trump if it is a buy of the yes side of the Trump contract (PRES-2024-DJT) or the no side of the Harris contract (PRES-2024-KH), and long Harris for the opposite positions.\footnote{Direction is measured from buy-side flows throughout. Every contract has a yes and a no side, so a position against a candidate appears as a buy of the opposite side. Sells close existing positions and enter the analysis through profit and loss. Buy-side gross flow is the natural measure of where traders committed new money.} Table~\ref{tab:direction_cohort} reports each cohort's dollar-weighted split on the two main presidential contracts.

\begin{table}[ht]
\centering
\caption{Direction of trades on the main presidential market by cohort, dollar-weighted.}
\label{tab:direction_cohort}
\begin{tabular}{lrrr}
\toprule
Cohort & \% long Trump & \% long Harris & Volume (\$M) \\
\midrule
=1                  & 27.4 & 72.6 &   3.4 \\
2--10               & 56.2 & 43.8 &  34.4 \\
11--100             & 58.7 & 41.3 &  38.7 \\
101--5{,}000        & 68.2 & 31.8 &  33.0 \\
\textgreater 5{,}000 & 63.5 & 36.5 &  59.1 \\
\bottomrule
\end{tabular}
\end{table}

One-trade users, whose behavior most closely resembles a single expressive vote, committed 72.6\% of their volume long Harris, even though Harris trailed in market prices throughout the sample and lost the election. The direction reverses as trade intensity rises. The most active traders were 64\% long Trump.

The casual tilt toward Harris has two possible explanations. One is expression, a single trade on the preferred candidate, a vote cast with money. The other is the favorite-longshot bias documented in betting markets for over a century, in which casual participants overpay for low-probability, high-payout positions \citep{snowberg-wolfers2010}. Harris contracts traded near 40 cents for most of the sample, and the one-trade cohort bought, across the two main contracts, at prices implying a 43.8\% win rate (Section~\ref{subsec:profit}). Section~\ref{subsec:geo} shows the casual Harris tilt is platform-wide rather than concentrated in Democratic areas, which is at least as consistent with the longshot explanation as with the expressive one. For our question the two explanations are equivalent. Both are motives unrelated to forecasting, neither depends on the trader's demographic category, and the incentives account requires only that such motives, whatever their mix, be concentrated among casual traders and absent at the price-setting margin, which is what the data shows.

The pattern is not specific to the presidential contract. We classify the yes side of every election market as Republican-aligned or Democrat-aligned from the contract definition (the classification, described in Section~\ref{subsec:geo}, covers 92.6\% of buy-side dollars) and count a no-side buy toward the opposite party. A position aligned with a party pays if that party wins, and we label such positions \emph{long Republican} and \emph{long Democrat}, extending the long-Trump terminology to the full election complex. Table~\ref{tab:direction_cohort_all} reports each cohort's split across the full election complex, state and popular-vote contracts, Senate seats, chamber control, and margin markets as well as the presidency. The gradient is the same. One-trade users placed 72.3\% of their election-market volume on Democrat-aligned positions. The long-Republican share rises monotonically with trade intensity to 65.9\% for the heaviest cohort. The expressive cohort took the losing side across the entire election complex.

\begin{table}[ht]
\centering
\caption{Direction of trades across all party-classifiable election markets by cohort, dollar-weighted.}
\label{tab:direction_cohort_all}
\begin{tabular}{lrrr}
\toprule
Cohort & \% long Republican & \% long Democrat & Volume (\$M) \\
\midrule
=1                  & 27.7 & 72.3 &   3.8 \\
2--10               & 53.4 & 46.6 &  42.0 \\
11--100             & 53.5 & 46.5 &  69.0 \\
101--5{,}000        & 60.7 & 39.3 &  76.7 \\
\textgreater 5{,}000 & 65.9 & 34.1 & 122.1 \\
\bottomrule
\end{tabular}
\end{table}

The cohort whose behavior most resembles political expression was the cohort most heavily aligned with the loser. Heavy traders, whose behavior cannot plausibly be described as expression, were aligned with the winner. The dollar-weighted aggregate is also dominated by the heavy cohort. Even if every casual trader had traded purely expressively, the market price would still have been determined by the small number of high-volume traders whose behavior is consistent with profit-seeking.

\paragraph{Sex and age.} The dimensions on which Kalshi is least representative are sex and age, and in 2024 both pointed toward Trump. Men voted 55 to 43 for Trump and women 53 to 45 for Harris, and men under 30 split evenly while women under 30 voted 61 to 38 for Harris.\footnote{National exit poll conducted by Edison Research for the National Election Pool, as published by NBC News. The sex-by-age cells used for the benchmark are, for men, 18--29 (49--48 Trump), 30--44 (52--45), 45--64 (59--39), and 65 and over (56--43). For women they are 18--29 (38--61), 30--44 (41--56), 45--64 (50--49), and 65 and over (46--53).} An expressive account therefore predicts that a platform on which male traders supply 97\% of volume tilts toward Trump through its composition. Table~\ref{tab:direction_sex_age} shows the opposite in the two least active cohorts. Male traders committed a smaller share of their volume long Republican than female traders did. The male-minus-female gap in dollar-weighted direction is $-11$ percentage points among one-trade users ($t=-2.3$), $-5$ percentage points in the 2 to 10 cohort ($t=-2.2$), and indistinguishable from zero above that. Within-cohort regressions holding zip partisanship and age fixed put the male coefficient at $-0.10$ and $-0.05$ in the two casual cohorts and $-0.03$ in the 101+ cohort. Age runs the wrong way for expression as well. Traders under 30, the group that voted most Democratic, committed a larger share of their volume long Republican than traders over 45 in every cohort except the one-trade cohort. The last two columns make the point in one number. Assigning each trader the two-party Trump vote share of his or her sex-age cell gives an expressive benchmark between 52.5\% and 53.7\% long Republican in every cohort, because every cohort consists mostly of young male traders. Actual direction runs from 27\% to 63\% across cohorts. Demographic composition predicts a flat line. The data shows a gradient in trade intensity.

\begin{table}[ht]
\centering
\caption{Direction of trades by sex and age within cohort: share of classified election-market dollars committed long Republican. The expressive benchmark assigns each trader the two-party Trump share of his or her sex-age cell in the 2024 national exit poll and averages over the cohort's dollars. Sex is imputed and is missing for 20\% of users.}
\label{tab:direction_sex_age}
\small
\begin{tabular}{lccccccc}
\toprule
 & \multicolumn{2}{c}{Sex} & \multicolumn{3}{c}{Age} & \multicolumn{2}{c}{Cohort} \\
\cmidrule(lr){2-3}\cmidrule(lr){4-6}\cmidrule(lr){7-8}
Cohort & Female & Male & 18--29 & 30--44 & 45 and over & Benchmark & Actual \\
\midrule
=1       & 36.3 & 25.4 & 24.4 & 26.4 & 34.5 & 52.5 & 27.1 \\
2--10    & 56.0 & 50.8 & 58.3 & 54.1 & 42.8 & 52.6 & 51.3 \\
11--100  & 51.7 & 53.1 & 56.0 & 52.6 & 51.1 & 53.2 & 53.0 \\
101+     & 67.2 & 62.6 & 65.4 & 64.5 & 56.3 & 53.7 & 62.8 \\
\midrule
\multicolumn{8}{l}{\textit{Gaps in percentage points, bootstrap $t$ in parentheses}} \\
 & \multicolumn{2}{c}{Male $-$ Female} & \multicolumn{3}{c}{(18--29) $-$ (45 and over)} & & \\
=1       & \multicolumn{2}{c}{$-10.9$ ($-2.3$)} & \multicolumn{3}{c}{$-10.0$ ($-3.0$)} & & \\
2--10    & \multicolumn{2}{c}{$-5.2$ ($-2.2$)}  & \multicolumn{3}{c}{$+15.5$ ($3.0$)}  & & \\
11--100  & \multicolumn{2}{c}{$+1.3$ ($0.5$)}   & \multicolumn{3}{c}{$+5.0$ ($1.3$)}   & & \\
101+     & \multicolumn{2}{c}{$-4.5$ ($-1.1$)}  & \multicolumn{3}{c}{$+9.1$ ($1.6$)}   & & \\
\bottomrule
\end{tabular}
\end{table}

\subsection{Partisan attribution of direction and outcomes}\label{subsec:geo}

This subsection asks the central behavioral question directly, whether a trader's political affiliation influences the direction of the trader's positions, and whether any such influence shows up in outcomes. The direction results say nothing yet about politics. A Harris-heavy casual cohort could be Democrats backing their hopes or Republicans hedging their fears. To bring political preference into the analysis we exploit geographic variation. If trading is preference-driven, the partisan composition of a trader's home district should predict the direction of the trader's positions, and, because 2024 was a Republican victory, should predict outcomes as well. If trading is incentive-driven, both relationships should weaken with trade intensity and vanish for the cohort that sets prices.

Rather than partitioning zip codes into discrete bins, we attribute each trader fractionally to the two parties according to the home zip code's vote. A trader living in a zip code that voted 80--20 for Trump counts as 0.80 Republican and 0.20 Democrat. Formally, let $w_i$ denote the two-party Republican share of the 2024 presidential vote in trader $i$'s home zip code, allocated from precinct returns, and let $y_i$ be a trader-level outcome. The Republican-attributed and Democrat-attributed means of the outcome are
\[
\bar y_R \;=\; \frac{\sum_i w_i\, y_i}{\sum_i w_i}, \qquad \bar y_D \;=\; \frac{\sum_i (1-w_i)\, y_i}{\sum_i (1-w_i)},
\]
where the sums run over traders. In the limiting case in which every zip code votes entirely for one party, $w_i$ is 0 or 1, $\bar y_R$ is the average outcome of the traders who live in Republican zip codes, and $\bar y_D$ is the average outcome of the traders who live in Democratic zip codes. When zip codes are mixed, each trader contributes to both means in proportion to the vote of the trader's zip code. Because the sums run over traders, a zip code with more traders receives proportionally more weight. The dollar-weighted version multiplies each trader's weight by the trader's dollar volume. Under this attribution, $\bar y_R - \bar y_D$ estimates the difference in mean behavior between Republican and Democratic residents of the traders' communities. The weights apply to aggregates, and no individual trader is assigned a party. We compute the difference within each trade-count cohort, grouping the two most active cohorts as 101+, and test it with a $t$-statistic.\footnote{The two attributed populations share the same underlying traders, so a two-sample test is inappropriate. We compute the variance of the trader-level weighted-mean difference by the delta method over trader-level observations. A nonparametric bootstrap over traders reproduces the delta-method standard errors to two decimal places. For the dollar-weighted gap, whose sampling variation is dominated by a small number of high-volume traders, we report $t$-statistics from the trader-level bootstrap directly. The 2024 vote allocation covers 99.2\% of users. Repeating the analysis with 2020 ZCTA vote shares as weights leaves the results in Tables~\ref{tab:attrib_dir} and \ref{tab:attrib_out} nearly unchanged. Weighting instead by the zip code's donor lean, the Trump share of its 2024 presidential donors in the DIME database, gives the same result. The dollar-weighted gap runs 7.8, 8.6, and 4.9 percentage points across the casual cohorts and 0.4 percentage points ($t=0.3$) for the 101+ cohort.} Direction $y_i$ is trader $i$'s long-Republican share, the share of buy-side dollars committed to Republican-aligned positions across all election markets (Section~\ref{subsec:direction}). We classify the yes side of each market as Republican- or Democrat-aligned from the contract definition. The classified markets are the Trump and Harris contracts, party-of-winner contracts for the presidency, popular vote, individual states, Senate seats, and chamber control, and margin-of-victory contracts aligned with the party that must win for the contract to pay. A no-side buy counts toward the opposite party. The classification covers 395 markets and 92.6\% of all buy-side dollars. The excluded remainder is non-partisan (closest-state contracts, split-ticket and mixed-control combinations, swing-state counts). Outcomes are the indicator for finishing the sample profitable and winsorized ROI, computed over all markets the user traded.

\paragraph{Direction.} Table~\ref{tab:attrib_dir} shows three results. First, partisan attribution moves the direction of trading far less than the preference account requires. Averaged over all users, the Republican-attributed population committed 60.3\% of its classified volume to Republican-aligned positions and the Democrat-attributed population 59.7\%, both firmly on the same side of the market. Second, the one-trade cohort is aligned with the Democratic side under either attribution (31.6\% versus 24.7\% of volume long Republican). The expressive cohort's tilt toward the losing side is a platform-wide phenomenon, not a partisan-geography phenomenon, and even the Republican-attributed one-trade population traded against the preference attributed to it. Third, and most important, the dollar-weighted partisan gap declines with trade intensity in both magnitude and statistical strength, 6.8 percentage points ($t=10.1$) for one-trade users, 6.9 ($t=7.4$) and 4.5 ($t=4.1$) for the middle cohorts, and vanishes for the 101+ cohort ($-0.2$ percentage points, $t=-0.1$).\footnote{Within the 101+ cohort, the gap for users with more than 5{,}000 trades is $-0.9$ percentage points ($t=-0.5$). If anything, Democrat-attributed heavy-trader volume leans slightly more Republican. Binning traders by dollar volume or by shares bought instead of trade count gives the same result. The dollar-weighted gap is $-0.7$ percentage points ($t=-0.5$) for the 9{,}922 traders with more than \$10{,}000 of volume and $-1.0$ percentage points ($t=-0.6$) for the 6{,}012 traders who bought more than 10{,}000 election-market shares.}
\begin{defbox}
Among the traders who dominate dollar volume and set the market price, volume from Republican- and Democrat-attributed traders is directionally indistinguishable.
\end{defbox}
The trader-level gap falls less, from roughly 8 percentage points among casual cohorts to 5.6 percentage points ($t=20.0$) for the 101+ cohort, because counting each trader equally gives weight to the cohort's many moderate-volume members. The dollar-weighted column, the margin that moves prices, is the one that converges to zero. The trader-level gaps remain statistically significant in every cohort. This does not contradict the framework of Section~\ref{subsec:framework}. The price is a dollar-weighted average of beliefs, so the framework requires the partisan gap to be zero only when traders are weighted by dollar volume. A partisan gap among traders who account for little of the volume does not move the price.

\begin{table}[ht]
\centering
\caption{Direction of election-market trades under fractional partisan attribution, by cohort. R and D columns are the $w$- and $(1-w)$-weighted means over the same traders, where $w$ is the 2024 two-party Republican vote share of the trader's home zip code. They are attribution-weighted means, not zip-code bins. ``\$-weighted'' aggregates dollars, ``trader-level'' averages each trader's long-Republican share. Gap is R $-$ D.}
\label{tab:attrib_dir}
\begin{tabular}{lrrrrrrrrrr}
\toprule
 & & \multicolumn{4}{c}{\% long Republican (\$-weighted)} & \multicolumn{4}{c}{\% long Republican (trader-level)} \\
\cmidrule(lr){3-6}\cmidrule(lr){7-10}
Cohort & Users & R & D & Gap & $t$ & R & D & Gap & $t$ \\
\midrule
=1            &  24{,}673 & 31.6 & 24.7 &  6.8 & 10.1 & 35.3 & 27.0 & 8.3 & 36.9 \\
2--10         & 111{,}170 & 57.4 & 50.5 &  6.9 &  7.4 & 67.2 & 59.1 & 8.2 & 81.7 \\
11--100       &  46{,}974 & 56.2 & 51.6 &  4.5 &  4.1 & 65.4 & 58.7 & 6.7 & 48.6 \\
101+          &  10{,}235 & 64.0 & 64.1 & $-0.2$ & $-0.1$ & 65.1 & 59.4 & 5.6 & 20.0 \\
\midrule
All           & 193{,}052 & 60.3 & 59.7 &  0.6 &  0.5 & 62.8 & 54.8 & 8.0 & 105.3 \\
\bottomrule
\end{tabular}
\\[0.25em]
\footnotesize Users are traders with a matched zip code and at least one buy on a party-classifiable election market. Dollar-weighted $t$-statistics are based on a bootstrap over traders, trader-level $t$-statistics on influence-function standard errors.
\end{table}

\paragraph{The zip-code view.} An equivalent way to present Table~\ref{tab:attrib_dir} is to compare each zip code's vote directly with what its traders did. Figure~\ref{fig:zip_vote_bet} plots the Republican share of a zip code's buy-side dollars against its two-party Republican vote share for zip codes with at least twenty matched traders, grouped into twenty bins by vote share. Trade direction is compressed toward the platform-wide Republican tilt. The trader-weighted slope is 0.39 (0.41 unweighted), in 77\% of zip codes the Republican share of volume exceeds the Republican share of the vote, and the mean gap between trading share and vote share is 15.5 percentage points. The compression is one-sided. Zip codes that voted under 40\% Republican committed 60\% of their volume long Republican, while zip codes over 60\% Republican committed 66\%, almost exactly their vote. Republican communities traded in line with their vote, and Democratic communities traded against theirs.\footnote{The zip-code view cannot support the cohort split that carries this section. Too few zip codes have twenty or more traders from a single cohort. The fractional attribution, which uses all matched traders, is the primary design for that reason.} The same comparison tests whether the uneven geography of participation distorted the platform's direction, which is the covariance that the own-type model of Appendix~\ref{app:model} requires to be nonzero for participation to distort the price (Proposition~\ref{prop:owntype}). Reweighting so that every location counts in proportion to its population moves the long-Republican share of these zip codes by less than two percentage points, from 58.5\% to 60.2\%. The geographic skew in who trades therefore did not move the platform's direction.

\begin{figure}[ht]
  \centering
  \includegraphics[width=0.75\textwidth]{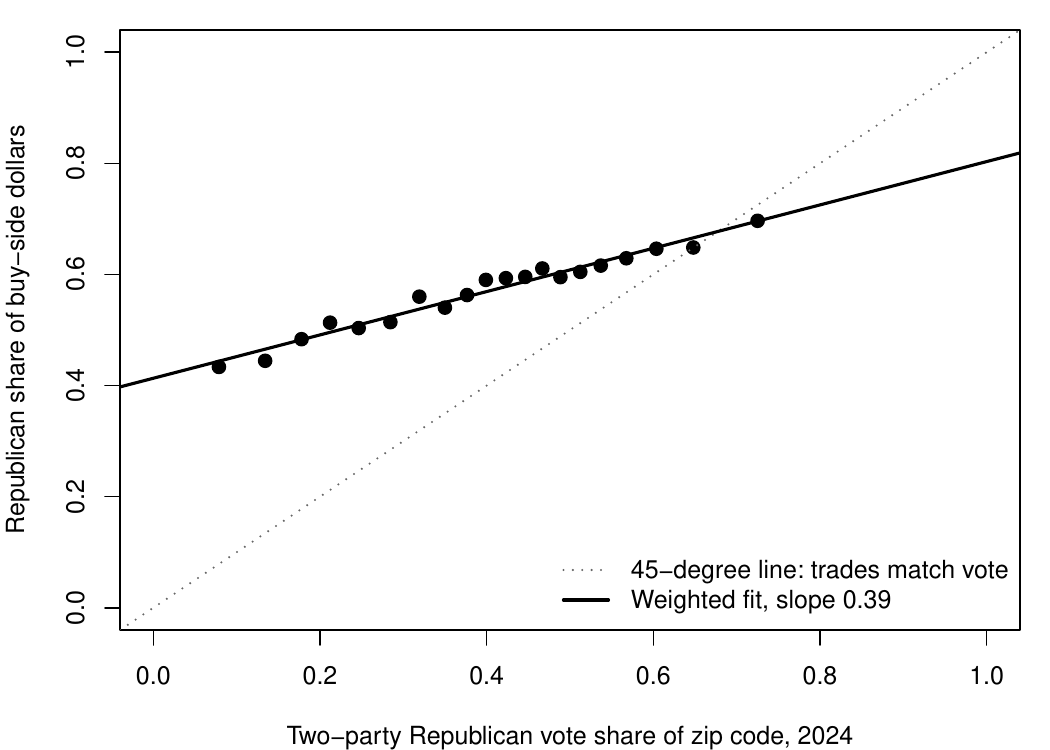}
  \caption{Trade direction against the vote, Kalshi 2024 election markets. Each zip code's Republican share of buy-side dollars across all election markets is plotted against its two-party Republican share of the 2024 presidential vote. Zip codes with at least twenty matched traders are sorted by Republican vote share and grouped into twenty bins with equal numbers of traders, and each point is a bin average. The dotted line marks trades matching the vote. The solid line is the trader-weighted least-squares fit across zip codes. In Democratic-voting zip codes the Republican share of volume is far above the Republican share of the vote. In Republican-voting zip codes the two are close.}
  \label{fig:zip_vote_bet}
\end{figure}

A regression that includes both vote measures gives a related result. When the 2020 and 2024 zip code vote shares are entered together with cohort indicators, the coefficient on the 2024 vote is $+1.24$ and the coefficient on the 2020 vote is $-0.69$, both highly significant. Trade direction, like profit (Section~\ref{subsec:profit}), tracks where the vote moved between 2020 and 2024 rather than long-standing partisanship. Traders in zip codes that moved toward Trump may have acted on local information, or those zip codes may be where the casual traders who bought Trump contracts live. The absence of a partisan gap among heavy traders favors the second explanation. Under either explanation, the result is inconsistent with an account in which long-standing partisan loyalty determines the direction of trades.

\paragraph{Outcomes.} Table~\ref{tab:attrib_out} repeats the exercise for profitability. Among casual traders the Republican-attributed population did better. One-trade Republican-attributed users finished profitable 6.3 percentage points more often than their Democrat-attributed counterparts ($t=30.4$), with a 0.11 gap in mean ROI. This is an ex post directional effect, examined closely in Section~\ref{subsec:profit}. Casual traders in Republican communities leaned slightly more toward Trump, and Trump won. The gap shrinks as trade intensity rises and is economically zero for the 101+ cohort, $-0.8$ percentage points on the profitable share and $+0.003$ on ROI, signs pointing in opposite directions and both negligible next to the casual-cohort gaps.\footnote{Within the 101+ cohort, the profitable-share gap for users with more than 5{,}000 trades is $-5.8$ percentage points ($t=-2.9$). Among the heaviest traders, Democrat-attributed users were, if anything, slightly more likely to finish profitable.} The partisan gap in outcomes appears where the partisan gap in direction appears, among the casual cohorts. For the traders whose volume determines the price, zip-level partisanship predicts neither how they traded nor what they earned.

\begin{table}[ht]
\centering
\caption{Trading outcomes under fractional partisan attribution, by cohort. Outcomes are over all markets, as in Section~\ref{subsec:profit}. Gap is R $-$ D.}
\label{tab:attrib_out}
\begin{tabular}{lrrrrrrrr}
\toprule
 & \multicolumn{4}{c}{\% profitable} & \multicolumn{4}{c}{Mean ROI (winsorized)} \\
\cmidrule(lr){2-5}\cmidrule(lr){6-9}
Cohort & R & D & Gap & $t$ & R & D & Gap & $t$ \\
\midrule
=1            & 28.8 & 22.5 &  6.3 & 30.4 & $-0.510$ & $-0.619$ & $ 0.108$ & 30.5 \\
2--10         & 53.9 & 47.2 &  6.7 & 60.4 & $ 0.052$ & $-0.035$ & $ 0.087$ & 63.5 \\
11--100       & 34.5 & 31.0 &  3.6 & 21.8 & $-0.081$ & $-0.114$ & $ 0.034$ & 28.6 \\
101+          & 29.8 & 30.6 & $-0.8$ & $-2.3$ & $-0.051$ & $-0.054$ & $ 0.003$ & 2.6 \\
\midrule
All           & 44.9 & 39.0 &  5.9 & 70.6 & $-0.056$ & $-0.134$ & $ 0.078$ & 77.3 \\
\bottomrule
\end{tabular}
\\[0.25em]
\footnotesize Computed over all users in each cohort with a matched home zip code (sample sizes differ slightly from Table~\ref{tab:attrib_dir}, which additionally requires a buy on a party-classifiable election market).
\end{table}

\paragraph{Donors versus voters.} The DIME contribution data allows a further question, whether the direction of a zip code's traders is better predicted by how the zip code voted or by which party its campaign donors supported. Donors are the residents who are most politically engaged and who commit money to politics. If the partisan gap among casual traders reflects political engagement carried into the market, the donor measure should predict direction beyond what the vote predicts. We estimate
\[
S_i = \alpha + \beta_1 \, \text{VoteShare}_{z(i)} + \beta_2 \, \text{DonorShare}_{z(i)} + \gamma_{c(i)} + \varepsilon_i,
\]
where $S_i$ is trader $i$'s long-Republican share, $\text{VoteShare}_{z(i)}$ is the 2024 two-party Republican vote share of the trader's home zip code, $\text{DonorShare}_{z(i)}$ is the Trump share of the zip code's 2024 presidential campaign donors, and $\gamma_{c(i)}$ is an indicator for the trader's cohort. Table~\ref{tab:donor_vote} reports the estimates. Entered one at a time and without cohort indicators, the two measures predict direction about equally. A trader's long-Republican share rises by 0.55 when the vote share changes from 0 to 1 (column 1) and by 0.54 for the same change in the donor share (column 2). Entered together with cohort indicators (column 3), the coefficient on the donor share is 0.37 and the coefficient on the vote share is 0.16. The two measures have a correlation of 0.94, so the comparison is imprecise, but it indicates that traders resemble the donors of their zip code more than its voters. This describes who chooses to trade. It does not show that prices reflect the politics of traders. Column (4) weights traders by dollar volume, and neither coefficient is statistically distinguishable from zero. Demographic and geographic variables together explain little of the variation in direction across traders (Section~\ref{subsec:lpm}). The number of donors per 1,000 residents, a measure of how politically engaged a zip code is, has almost no relationship with the direction of trades, and none with profit (Section~\ref{subsec:profit}). The Republican share of trading volume correlates 0.34 with the Republican vote share and 0.36 with the Republican donor share, based on trader location. However partisanship is measured, it is weakly related to the direction of a zip code's trades.

\begin{table}[ht]
\centering
\caption{User-level regressions of the long-Republican share of classified election-market dollars on the 2024 Republican vote share and the 2024 Republican donor share of the home zip code. The donor share is the Trump share of the zip code's 2024 presidential campaign donors in the DIME database, defined for zip codes with at least ten such donors. Column (4) weights observations by classified dollar volume. Standard errors, in parentheses, are heteroskedasticity-robust.}
\label{tab:donor_vote}
\begin{tabular}{lcccc}
\toprule
 & (1) & (2) & (3) & (4) \\
\midrule
Zip Republican vote share (2024)  & $0.550$\rlap{$^{***}$} & & $0.164$\rlap{$^{***}$} & $-0.085$ \\
                                  & $(0.005)$ & & $(0.015)$ & $(0.214)$ \\
Zip Republican donor share (2024) & & $0.537$\rlap{$^{***}$} & $0.373$\rlap{$^{***}$} & $0.311$ \\
                                  & & $(0.005)$ & $(0.015)$ & $(0.201)$ \\
Cohort indicators                 & no & no & yes & yes \\
Weighted by dollar volume         & no & no & no & yes \\
Observations                      & 192{,}518 & 192{,}518 & 192{,}518 & 192{,}518 \\
$R^2$                             & 0.055 & 0.057 & 0.111 & 0.064 \\
\bottomrule
\end{tabular}
\\[0.25em]
\footnotesize $^{***}\,p<0.001$.
\end{table}

\paragraph{Attenuation.} The attribution treats a zip code's vote share as the probability that a trader who lives there is a Republican. Traders are a selected group within their zip codes, so the measure could understate the partisanship of individual traders, and could understate it more for heavy traders than for casual traders. Appendix~\ref{app:robust} reports three checks. First, in zip codes that voted more than 75\% or less than 25\% Republican, the difference in dollar-weighted direction between residents of Republican and Democratic zip codes is 31 percentage points for one-trade users and 37 percentage points for the 2 to 10 cohort, and it is $-2.5$ percentage points ($t=-0.4$) for the 101+ cohort. Second, the home zip codes of heavy traders are no less informative about party than those of casual traders. Third, in the PredictIt data, where registered party is observed, the vote share of a trader's area closely tracks the share of traders registered Republican. The absence of a partisan gap among heavy traders is therefore not produced by a weaker measure of their politics.

\paragraph{Two-sided and one-sided traders.} A different objection is that the heavy cohort's zero gap is mechanical. Traders who supply liquidity hold positions on both sides of a contract, and the direction of such positions is inventory rather than opinion. We split traders by whether at least 10\% of their main-contract dollars sat on the minority side of their own book. No one-trade user is two-sided by construction. The two-sided share of volume rises from 17\% in the 2 to 10 cohort to 38\%, 52\%, and 92\% for traders with more than 5,000 trades (Table~\ref{tab:sidedness}). The split supports part of the objection. Among traders with more than 100 trades, the two-sided group, which supplies four-fifths of the cohort's volume, shows no partisan gap (attribution gap $-1.3$ percentage points, $t=-0.7$; gap between zip codes that voted more than 75\% and less than 25\% Republican $-7$ percentage points, $t=-0.9$) and paid fair prices on the main contracts ($z=-3.0$, Section~\ref{subsec:profit}). The one-sided heavy traders show a gap of 4.8 percentage points ($t=2.2$), 20 percentage points in homogeneous zip codes, and were on the winning side of the election ($z=+25.8$), a result that one outcome cannot distinguish from luck. The same split within the 11 to 100 cohort gives 8 and 33 percentage points for one-sided traders and zero for two-sided. The partisan gap among one-sided traders falls only in the most active cohort, from 31 to 37 percentage points in the casual cohorts and 33 percentage points in the 11 to 100 cohort to 20 percentage points for the 101+ cohort on the homogeneous-zip measure, so the partisan gap among directional traders narrows with trade intensity, but less completely than the cohort tables alone suggest. At the same time, the volume that sets the price comes overwhelmingly from two-sided capital whose positions carry no partisan signal. Trade count as such does not remove partisanship from the price. As traders scale up they shift from taking a position on the outcome to trading the price, and that shift, which is the profit motive at work, removes it.

\begin{table}[ht]
\centering
\caption{Two-sided and one-sided traders. A trader is two-sided if at least 10\% of his or her main-contract dollars sat on the minority side of the trader's own positions. Direction is the share of classified election-market dollars committed long Republican. Gaps are Republican-attributed minus Democrat-attributed (fractional attribution) and Red-zip minus Blue-zip residents in homogeneous zip codes, with bootstrap $t$ in parentheses. $z$ is the calibration score of Section~\ref{subsec:profit} on the two main presidential contracts.}
\label{tab:sidedness}
\footnotesize
\begin{tabular}{llrrrrrr}
\toprule
Cohort & Traders & Users & \$M & \% long R & Attrib. gap & Homog.-zip gap & $z$ \\
\midrule
11--100 & one-sided & 31,074 & 36.7  & 51.6 & $8.0$ ($5.4$)   & $32.5$ ($4.4$)  & -- \\
11--100 & two-sided & 8,868  & 22.5  & 59.5 & $-0.8$ ($-0.4$) & $-11.4$ ($-1.4$) & -- \\
101+    & one-sided & 6,136  & 40.7  & 62.2 & $4.8$ ($2.2$)   & $19.8$ ($2.4$)  & $+25.8$ \\
101+    & two-sided & 2,227  & 142.2 & 65.4 & $-1.3$ ($-0.7$) & $-7.4$ ($-0.9$) & $-3.0$ \\
\midrule
\multicolumn{8}{l}{Two-sided share of cohort volume: =1, 0\%; 2--10, 17\%; 11--100, 38\%; 101--5,000, 52\%; over 5,000, 92\%.} \\
\bottomrule
\end{tabular}
\end{table}

The graduated pattern is what the incentives mechanism predicts. Partisan geography leaves a visible mark on the behavior and the outcomes of the traders who trade least, the cohort whose single trade most resembles expression, a fading mark on directional traders as they scale up, and none on the two-sided capital that supplies most of the price-setting volume. Prices paid show the same pattern in weaker form. Volume-weighted average prices for the Trump contract differ by under two cents across the partisan spectrum, but because all traders transact against a single central order book, price convergence is close to mechanical. The informative margins are direction and outcomes, and this section rests on them for that reason.

\subsection{Regression evidence}\label{subsec:lpm}

The attribution analysis can be summarized in regression form, which serves two purposes. It quantifies how much of trading behavior demographics and geography explain, and it verifies that the partisan gradient survives demographic controls. For each user $i$ we take the long-Republican share $S_i \in [0,1]$ of classified election-market dollars defined in Section~\ref{subsec:geo} and estimate
\[
S_i = \alpha + \beta_1 \, \text{RepShare24}_{z(i)} + \beta_2 \, \text{Male}_i + \beta_3 \, \text{Age}_i + \gamma_{c(i)} + \varepsilon_i,
\]
where $\text{RepShare24}_{z(i)}$ is the two-party Republican share of the 2024 presidential vote in user $i$'s home zip code and $\gamma_{c(i)}$ is a cohort fixed effect. Observations are weighted by classified dollar volume, the price-relevant margin.\footnote{Unweighted estimation raises the partisanship coefficient to 0.56 and the $R^2$ to 0.12. Trader counts weight the casual users, so the unweighted slope is dominated by the cohorts in which the partisan tilt lives.}

In Table~\ref{tab:lpm_user} the coefficient on RepShare24 is the empirical counterpart of the projection parameter $\rho$ of Section~\ref{subsec:framework}, estimated on the geographic dimension of type and weighted by the volume that carries the price. Table~\ref{tab:lpm_user} shows two results. First, explanatory power is minimal. Demographics and home-zip partisanship together explain under 1\% of the dollar-weighted variation in direction, and adding cohort fixed effects raises the $R^2$ only to 0.078. More than nine-tenths of the variation in the direction traders take is unexplained by who they are, where they live, and how much they trade. Second, the within-cohort rows reproduce the attribution gradient in regression form. The partisanship coefficient is 0.51 and 0.56 in the two casual cohorts, falls to 0.30 in the middle cohort, and is 0.08 for the 101+ cohort, where it is not statistically distinguishable from zero (standard error 0.07, bootstrap 95\% confidence interval $-0.03$ to 0.24). That interval excludes the coefficients of the two casual cohorts, 0.51 and 0.56. A 10-percentage-point increase in home-zip Republican vote share is associated with an increase of roughly 5 percentage points in a casual trader's long-Republican share. For a trader in the 101+ cohort the estimated increase is 0.8 percentage points.\footnote{Volume in the 101+ cohort is concentrated in a small number of traders. Standard errors that treat the dollar weights as precision weights are therefore too small, by a factor of about five for this cohort. We report heteroskedasticity-robust standard errors, and a bootstrap over traders gives nearly the same values. Within the cohort with more than 5,000 trades the coefficient is 0.13 (standard error 0.14). The dollar-weighted attribution gap for the same cohort is $-0.9$ percentage points ($t=-0.5$, Section~\ref{subsec:geo}).}

\begin{table}[ht]
\centering
\caption{User-level regressions of the long-Republican share of classified election-market dollars on home-zip partisanship. Each cohort row is a separate regression with sex and age controls. The last row pools all traders and adds an indicator for each cohort (cohort fixed effects). Observations weighted by classified dollar volume. Standard errors are heteroskedasticity-robust.}
\label{tab:lpm_user}
\begin{tabular}{lrrrr}
\toprule
Cohort & Users & Mean $S_i$ & RepShare24 coef. & SE \\
\midrule
=1              &  20{,}112 & 0.268 & $0.513^{***}$ & (0.048) \\
2--10           &  90{,}601 & 0.513 & $0.561^{***}$ & (0.046) \\
11--100         &  38{,}734 & 0.529 & $0.300^{***}$ & (0.090) \\
101+            &   8{,}570 & 0.630 & $0.084$ & (0.071) \\
\midrule
All traders, cohort indicators & 158{,}017 & 0.598 & $0.308^{***}$ & (0.048) \\
\bottomrule
\end{tabular}
\\[0.25em]
\footnotesize $^{***}\,p<0.001$. Mean $S_i$ is dollar-weighted. Sex and age coefficients are small in every specification (male $-0.046$, age $-0.001$ in the pooled model). Pooled $R^2$: 0.008 without cohort fixed effects, 0.078 with.
\end{table}

\subsection{Trading outcomes and demographics}\label{subsec:profit}

The preceding subsections show that demographics are weak predictors of the direction of a trader's positions. Profit provides a second test. If the demographic homogeneity of the platform shaped market behavior in a way that mattered, it should surface in who wins and who loses money. The expressive version of the diversity critique predicts in particular that the demographic majority, young male traders, should drive prices and bear the corresponding gains or losses as a bloc. We test this directly.

We measure outcomes three ways. The first is return on investment, $\text{ROI}_i = \text{PnL}_i / \text{Volume}_i$, with PnL net of fees.\footnote{We winsorize ROI at the 1st and 99th percentiles in the regressions to limit the influence of a small number of extreme ratios. Descriptive tables report unwinsorized values.} The second is an indicator for finishing the sample profitable, and the third is the inverse hyperbolic sine of dollar PnL, which accommodates heavy tails without dropping observations.

\paragraph{Profitability is flat across the dimensions of non-diversity.} Table~\ref{tab:profit_group} reports outcomes by sex and age. On sex, the dimension on which Kalshi is most unrepresentative, profitability is flat. Male and female traders finish profitable at the same rate (41.6\% versus 41.7\%) with statistically indistinguishable returns. The one demographic gradient is age. The profitable share rises from 39.9\% among traders 18--29 to 51.5\% among traders 60 and over. The regression below shows that this difference reflects the side taken. Once the trader's own long-Republican share is included, the coefficient on age falls to zero (Table~\ref{tab:profit_reg}, column 4). None of these gaps is large, and none runs in the direction the expressive account requires.

\begin{table}[ht]
\centering
\caption{Trading outcomes by demographic group, Kalshi 2024. ROI is PnL net of fees divided by dollar volume. ``Profitable'' is the share of traders with positive PnL.}
\label{tab:profit_group}
\begin{tabular}{lrrrr}
\toprule
Group & Users & \% profitable & Mean ROI & Median ROI \\
\midrule
\multicolumn{5}{l}{\textit{Sex}}\\
\quad Male            & 139{,}523 & 41.6 & $-0.089$ & $-0.074$ \\
\quad Female          &  23{,}003 & 41.7 & $-0.116$ & $-0.090$ \\
\multicolumn{5}{l}{\textit{Age}}\\
\quad 18--29          & 102{,}212 & 39.9 & $-0.101$ & $-0.080$ \\
\quad 30--39          &  54{,}566 & 41.2 & $-0.105$ & $-0.084$ \\
\quad 40--49          &  25{,}974 & 45.3 & $-0.064$ & $-0.052$ \\
\quad 50--59          &  10{,}901 & 48.4 & $-0.039$ & $-0.019$ \\
\quad 60+             &   4{,}790 & 51.5 & $-0.012$ & $+0.015$ \\
\bottomrule
\end{tabular}
\end{table}

\paragraph{Regression evidence.} Table~\ref{tab:profit_reg} reports user-level regressions of each outcome on sex, age, and the 2024 Republican vote share of the home zip code. None of these characteristics predicts profit on its own. Each predicts profit only through the side the trader took. Because Trump won, any group that bought more Trump contracts earned more. Three comparisons across the columns of the table establish this.
\begin{enumerate}
\item Sex does not predict profit. In column (1) male traders earned 3.2 cents more per dollar traded than female traders. With cohort indicators added in column (2) the difference is 0.4 cents and not statistically significant. Male traders are concentrated in the more active cohorts, and the cohort indicators absorb the whole of the difference.
\item Location predicts profit only through direction. In column (2) a trader in a zip code that voted entirely Republican is predicted to earn 53 cents more per dollar than a trader in a zip code that voted entirely Democratic. Column (4) adds the trader's own long-Republican share. Its coefficient is 0.93, and the coefficient on the zip code's vote share falls from 0.53 to 0.02.
\item Age predicts profit only through direction. Older traders were more often profitable in columns (1) to (3). In column (4) the coefficient on age is zero.
\end{enumerate}
The share of the variation in ROI that the regression explains rises from 13\% in column (2) to 56\% in column (4).

The regressions also show that ROI rises with trade intensity, by 0.46 to 0.62 in every more active cohort relative to the one-trade cohort. These differences are not estimates of the effect of trading more, for three reasons. The cohorts are partly sorted by outcome, since a trader whose first position loses value may stop trading and a trader whose positions gain value has more capital to continue. About half of each difference is the side taken. One-trade users placed 73\% of their main-contract volume long Harris, and in column (4) the cohort differences fall to between 0.17 and 0.35. And ROI divides profit by dollar volume, which active traders raise by buying and selling the same contracts repeatedly, so their ROI is pulled toward zero. Our argument does not require that active traders are more skilled, only that the direction of their trades does not follow their politics, which Section~\ref{subsec:geo} shows.

\begin{table}[ht]
\centering
\caption{User-level regressions of trading outcomes on demographics. Male is an indicator (imputed sex, observations with missing sex omitted). Zip Republican share is the 2024 two-party Republican vote share of the home zip code. Own long-Republican share is the share of the trader's classified election-market dollars committed to Republican-aligned positions (Section~\ref{subsec:geo}). Cohorts are defined in Section~\ref{subsec:cohort}.}
\label{tab:profit_reg}
\begin{tabular}{lcccc}
\toprule
 & (1) ROI & (2) ROI & (3) Pr(profitable) & (4) ROI \\
\midrule
Male            & $0.032$\rlap{$^{***}$} & $0.004$ & $0.008$ & $0.005$ \\
                & $(0.004)$ & $(0.004)$ & $(0.003)$ & $(0.003)$ \\
Age (years)     & $0.0012$\rlap{$^{***}$} & $0.0018$\rlap{$^{***}$} & $0.0027$\rlap{$^{***}$} & $-0.0001$ \\
                & $(0.0001)$ & $(0.0001)$ & $(0.0001)$ & $(0.0001)$ \\
Zip Republican share (2024) & $0.566$\rlap{$^{***}$} & $0.531$\rlap{$^{***}$} & $0.400$\rlap{$^{***}$} & $0.022$\rlap{$^{***}$} \\
                & $(0.008)$ & $(0.007)$ & $(0.006)$ & $(0.005)$ \\
Own long-Republican share & & & & $0.929$\rlap{$^{***}$} \\
                & & & & $(0.002)$ \\
Cohort indicators & no & yes & yes & yes \\
Observations & 161{,}104 & 161{,}104 & 161{,}104 & 158{,}017 \\
$R^2$ & 0.034 & 0.127 & 0.071 & 0.557 \\
\bottomrule
\end{tabular}
\\[0.25em]
\footnotesize $^{**}\,p<0.01$, $^{***}\,p<0.001$. Standard errors in parentheses. Dependent variables: winsorized ROI (columns 1, 2, and 4) and an indicator for positive PnL (column 3). An inverse-hyperbolic-sine PnL specification (not shown) yields the same signs.
\end{table}

\paragraph{The 2020 and 2024 votes compared.} The share of traders who finished profitable rises from 33.7\% in zip codes that voted less than 40\% Republican in 2024 to 50.6\% in zip codes that voted more than 60\% Republican. This observation admits two possible explanations: Traders in Republican areas may be better forecasters, or they may simply have bought more Trump contracts in an election that Trump won. We separate these two possibilities with the following simple regression analysis. We enter the zip code's 2020 and 2024 Republican vote shares together in the regression of column (2). If traders from Republican areas were better forecasters, the long-standing partisanship of the zip code, measured by the 2020 vote, should predict profit on its own. According to our analysis, it does not. The coefficient on the 2024 share is $+1.20$ and the coefficient on the 2020 share is $-0.66$, both significant at the 0.1\% level. Holding the 2024 vote fixed, a trader in a zip code that was more Republican in 2020 earned less. Profit follows the shift in the vote toward Trump between 2020 and 2024, which is where the traders who bought Trump contracts lived, not the zip code's long-standing partisanship. Measures based on campaign donors give the same result. The number of donors per 1,000 residents, a measure of political engagement, has no statistically significant relationship with returns ($p=0.11$).

\paragraph{A calibration score.} ROI and the profitable share both depend on which candidate won. A trader who bought Trump contracts profited whether or not they overpaid or underpaid for the position. The calibration score asks a narrower question, whether a group's positions won as often as would be implied by the prices they paid. A position is a single buy trade on one of the two main presidential contracts, for which settlement is unambiguous. For position $i$, let $p_i$ be the price paid per contract, in dollars, and let $W_i$ equal 1 if the side bought paid out and 0 if it did not. A contract bought at price $p_i$ is priced by the market to win with probability $p_i$. If a group paid fair prices, the expected number of its positions that win is $\sum_i p_i$, with variance $\sum_i p_i(1-p_i)$ when positions are treated as independent.\footnote{Every position resolves on the same election, so the positions are not independent and the score should not be compared with standard normal critical values. We display it here in order to help in comparing groups, not as a formal test of any one group against zero.} The score is the number of positions that won minus the number the prices implied, divided by the standard deviation,
\[
z \;=\; \frac{\sum_i W_i - \sum_i p_i}{\sqrt{\sum_i p_i(1-p_i)}}.
\]
The numerator is the profit from buying one contract in each position, so the score is a standardized profit in which every position counts equally, whatever its size. A group with $z$-score near zero won as often as the prices it paid implied, and so paid fair prices on average. A group with $z$ below zero won less often than its prices implied, and so overpaid on average. For a group that bought both sides of the market, the election outcome largely cancels and the score measures whether the prices paid were fair at the time of purchase. For a group that bought mostly one side, the score mostly reflects whether that side won, regardless of price paid.

\begin{table}[ht]
\centering
\caption{Calibration $z$-score on the two main presidential contracts, by group. $z>0$ indicates a group won more often than the prices it paid implied. $z\approx 0$ indicates fair value. Positions is the number of buy trades. Implied is the average price paid per contract, in dollars, which is the win rate the prices implied. Realized is the share of positions that won.}
\label{tab:zscore}
\begin{tabular}{lrrrr}
\toprule
Group & Positions & Implied & Realized & $z$ \\
\midrule
All positions        & 367{,}615 & 0.489 & 0.521 & $+40.3$ \\
\multicolumn{5}{l}{\textit{By cohort}}\\
\quad =1             &  21{,}762 & 0.438 & 0.275 & $-50.6$ \\
\quad 2--10          & 174{,}445 & 0.489 & 0.549 & $+53.2$ \\
\quad 11--100        & 116{,}597 & 0.488 & 0.515 & $+19.4$ \\
\quad 101--5{,}000   &  36{,}730 & 0.510 & 0.556 & $+18.6$ \\
\quad \textgreater 5{,}000 & 18{,}081 & 0.518 & 0.508 & $-2.7$ \\
\multicolumn{5}{l}{\textit{By sex}}\\
\quad Male           & 263{,}598 & 0.487 & 0.513 & $+28.5$ \\
\quad Female         &  33{,}278 & 0.492 & 0.532 & $+15.9$ \\
\bottomrule
\end{tabular}
\end{table}

Table~\ref{tab:zscore} compares the cohorts on a single scale. Compared across cohorts, the score shows the pattern the incentives account predicts. One-trade users have a $z$ of $-50.6$. They paid prices implying a 43.8\% win rate and won 27.5\% of the time. The heavy-trader cohort that sets the market price has a $z$ of $-2.7$, the smallest in magnitude of any group. It paid prices implying a 51.8\% win rate and won 50.8\% of the time, within one percentage point of fair value. Section~\ref{subsec:framework} defined $B$ as the dollar-weighted average difference between traders' beliefs and the true probability. The score is its observable counterpart for one cohort. It compares the prices the cohort paid with how often its positions won. For traders with more than 5,000 trades, who account for 79\% of volume, the difference is one percentage point. Among all traders with more than 100 trades, the two-sided traders account for four-fifths of volume and have a score of $-3.0$ (implied win rate 51.4\%, realized win rate 50.7\%). The one-sided traders with more than 100 trades account for the remaining one-fifth and have a score of $+25.8$ (implied win rate 51.0\%, realized win rate 59.6\%). They bought mostly the winning side, and one election cannot distinguish that from luck (Section~\ref{subsec:geo}). Compared across demographic groups, the score shows no difference between the groups on which the platform is least representative. Male and female traders have similar positive scores, and the score does not vary systematically with age. Figure~\ref{fig:profit_demo} summarizes the two halves of the result. Profitability is flat across demographic groups and rises with trade intensity, and the calibration score is close to zero for the most active cohort.

\begin{figure}[h!]
  \centering
  \includegraphics[width=0.95\textwidth]{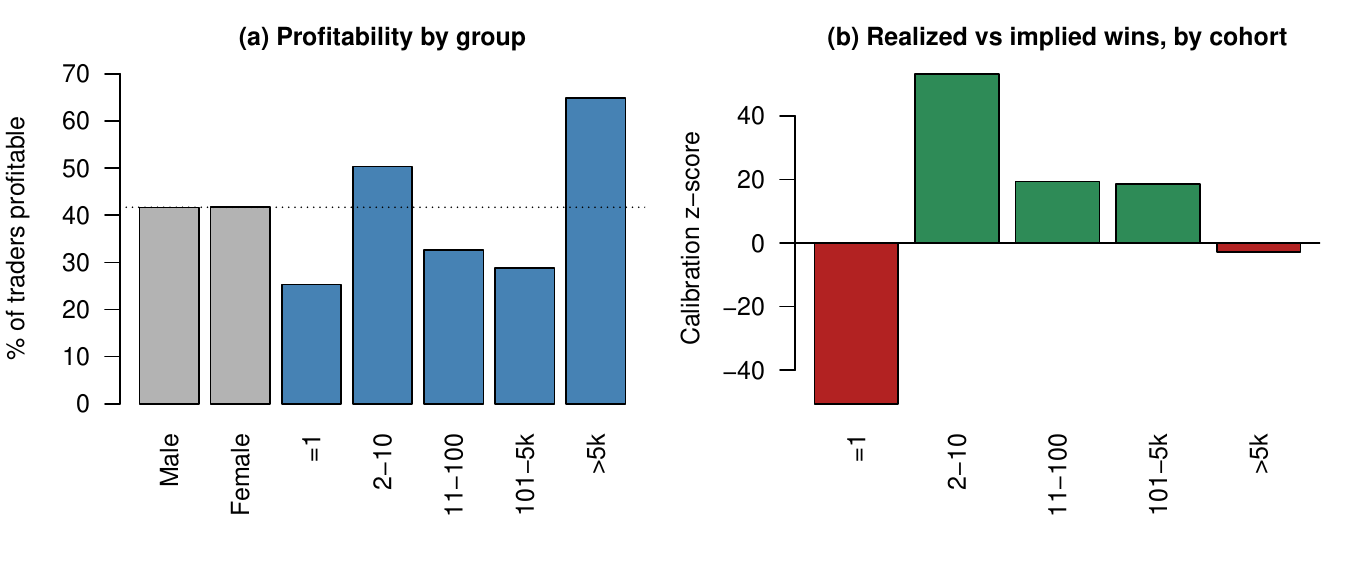}
  \caption{Trading outcomes on Kalshi's 2024 election markets. Left: share of traders finishing profitable across all election markets, by demographic group and by trade-count cohort (dashed line: platform average, 41.7\%). Right: calibration $z$-score on the two main presidential contracts, by cohort. Profitability is flat across demographics and rises with trade intensity. The heavy-trader cohort that sets prices is calibrated to fair value ($z\approx 0$).}
  \label{fig:profit_demo}
\end{figure}

The outcome evidence reinforces the direction evidence. Returns on Kalshi did not differ along demographic lines. They were highest in the most active cohort and lowest in the one-trade cohort, and the two-sided traders who account for most of the volume paid fair prices on average.

\subsection{External replication on PredictIt in 2016}\label{subsec:predictit}

The Kalshi sample is a single market in a single election, and its preference measures are ecological. To test external validity we replicate the user-level analysis on PredictIt's 2016 presidential market, using a separate proprietary dataset that includes registered party affiliation and verified federal donor status, individual-level preference measures sharper than zip-code attribution. The 2016 market is also a harder case for the incentives account. PredictIt priced Trump at 22 cents on the day before the election, so the market favored the candidate who lost. The data covers the entire election cycle, a longer period than the one month of Kalshi trading in the 2024 cycle. It records each trader's net purchases on each side of the presidential market, not individual trades or transaction prices. PredictIt is a research-oriented exchange operating under a CFTC no-action letter with an \$850 per-contract investment limit. Its smaller, politically engaged user base provides a useful contrast. Of the traders in the presidential market with a recorded sex, 92\% were men, and about 42\% of traders were under 30 in 2016, against 52\% female and 19\% under 30 in the 2016 electorate.\footnote{Sex is recorded for 55\% of these traders and is inferred from voter-file data by the data provider. Age in the data is age in July 2025, which we convert to age in November 2016. Electorate shares are from the national exit poll conducted by Edison Research for the National Election Pool, November 2016.} Their politics is also observed directly.\footnote{Voter-file records in the dataset confirm the engagement. Among the 62\% of accounts matched to a voter file, users who traded in 2016 voted in that year's general election at a 76\% rate, against 62\% for accounts that never traded, and turnout among traders varies only modestly by registered party (74\% to 80\%). For comparison, 61.4\% of the citizen voting-age population reported voting in 2016 in the Census Bureau's Current Population Survey. Turnout is unrelated to the direction of trading. Share-weighted long-Trump shares are 0.55 for 2016 voters and 0.53 for non-voters.} If incentives govern behavior, the qualitative pattern should hold. If their effect weakens where investment limits bind, the preference signal should be larger.

We restrict to the 962 users who purchased shares on the head-to-head Trump--Clinton market. Direction on PredictIt is measured in shares purchased, not dollars. Table~\ref{tab:predictit_party} reports the share-weighted long-Trump share by registered party.

\begin{table}[ht]
\centering
\caption{PredictIt 2016: long-Trump share by trader's registered party.}
\label{tab:predictit_party}
\begin{tabular}{lrrr}
\toprule
Registered party & Users & Share-weighted long-Trump share & Total shares \\
\midrule
Democrat (D)             & 178 & 0.358 & 694{,}405 \\
Republican (R)           &  61 & 0.576 & 104{,}567 \\
Independent (I)          &   4 & 0.399 &  44{,}794 \\
Other party code or none recorded & 719 & 0.592 & 2{,}769{,}592 \\
\bottomrule
\end{tabular}
\end{table}

Two patterns differ from Kalshi. Registered Republicans placed 58\% of their shares long Trump and registered Democrats 36\%, a gap of 22 percentage points, several times the partisan-attribution gaps inferred from Kalshi. Registered Democrats therefore placed most of their shares on Clinton, the candidate who lost. The comparison rests on the 239 traders registered with one of the two major parties. The other 723 traders are not registered with a major party in the data. Of these, 431 have no party recorded, 288 are recorded under another party code, and 4 are coded as independent. The 719 traders in the first two groups placed 59\% of their shares long Trump, close to the registered Republicans. A user-level regression (Table~\ref{tab:predictit_lpm}) confirms the pattern. Registered Democrats place 18 percentage points fewer of their shares long Trump than otherwise comparable users. Even here, preference and demographic variables together explain only 13\% of the variation in trading direction. The verified Clinton-donor coefficient is positive ($+0.31$). Traders who gave money to Clinton bought Trump more heavily. We interpret this as hedging. A Clinton supporter who buys Trump contracts insures against the outcome she fears, and hedging, like profit-seeking, is a financial rather than expressive motive. The donor coefficient is therefore further evidence against the expressive account.

\begin{table}[ht]
\centering
\caption{PredictIt 2016: user-level LPM of long-Trump share.}
\label{tab:predictit_lpm}
\begin{tabular}{lrr}
\toprule
Variable & Coefficient & Std. Error \\
\midrule
Trump donor (FEC)               &  $0.131$              & 0.144  \\
Clinton donor (FEC)             &  $0.310$\rlap{$^{**}$} & 0.088  \\
Registered Republican           &  $0.019$              & 0.054  \\
Registered Democrat             & $-0.181$\rlap{$^{***}$} & 0.025  \\
Male                            & $-0.068$\rlap{$^{**}$} & 0.021  \\
Age                             & $-0.005$\rlap{$^{***}$} & 0.001  \\
Constant                        &  $0.818$\rlap{$^{***}$} & 0.043  \\
\midrule
$N$ users                        &     962 & \\
$R^2$                             &   0.134 & \\
\bottomrule
\end{tabular}
\\[0.25em]
\footnotesize Dependent variable: per-user share of purchased shares long Trump on the Trump--Clinton head-to-head market, 2016. Observations weighted by total shares purchased. $^{**}\,p<0.01$, $^{***}\,p<0.001$.
\end{table}

The PredictIt results reinforce and qualify the main finding. The direction is consistent with Kalshi, but the preference signal is stronger. On a small exchange with binding investment limits and a self-selected, politically active user base, the incentives mechanism operates less completely. Even so, registered Republicans placed 42\% of their shares on Clinton to win. The framework of Section~\ref{subsec:framework} explains this contrast between markets.
\begin{defbox}
Accuracy of prediction market prices is determined by the concentration of volume from traders whose trading activity does not depend on their own politics.
\end{defbox}
PredictIt's \$850 limit caps every trader's weight $\lambda_i$, so the dollar-weighted trader distribution cannot separate from the head-count distribution. \cite{zitzewitz2026} studies the price consequences of the same limit. PredictIt's prices in 2024 were less accurate than those of the uncapped exchanges, as expected if the limit prevents profit-motivated traders from holding enough of the volume to set the price. Over the final month PredictIt's national Trump probability averaged 53\%, between the poll-based models at 50 to 51\% and the uncapped exchanges at 57 to 60\%, and its Brier score (0.219) sat between the two groups as well (Section~\ref{subsec:brier}). On Kalshi, heavy traders whose positions do not follow their politics hold most of the volume, so the partisan gap among casual traders does not reach the price. On PredictIt no trader is permitted to hold enough of the volume to have that effect. The Iowa Electronic Markets operate under a tighter limit, with a maximum total investment of \$500 per trader.

\subsection{Interpretation}\label{subsec:interpretation}

The findings of this section reject a strong preference-driven account of trading behavior and support a graduated incentives-driven account. Under the preference account, traders' positions should track the politics of their communities, demographic variables should be substantial predictors of direction, and the casual traders most plausibly trading expressively should have leaned toward Trump given the platform's demographic skew. The data supports none of these predictions. The partisan gap in dollar-weighted trade direction shrinks with trade intensity and vanishes for traders with more than 100 trades, whose zip-level partisanship predicts neither direction nor outcomes (Section~\ref{subsec:geo}). Demographics and home-zip partisanship together explain under 8\% of cross-user variation in direction, even with trade-intensity controls (Section~\ref{subsec:lpm}). The most expressive cohort was 73\% long Harris, against the winner, and had the lowest returns, while the heavy cohort was 64\% long Trump, had the highest returns, and bought at prices that proved fair (Sections~\ref{subsec:direction} and \ref{subsec:profit}). The PredictIt replication shows the same mechanism operating less completely under a binding investment cap, as an incentives account predicts and a demographic account does not.

Appendix~\ref{app:model} states the result generally. The price converges to the true probability if and only if the financially weighted belief error vanishes, a condition that does not mention the trader distribution, and the price is biased only where a nonzero selection gap meets a nonzero dependence of belief on type at the dollar-weighted margin. On Kalshi the selection gap is large, the association between a trader's politics and the direction of trades is not distinguishable from zero where volume concentrates, and on the geographic dimension the gap has the sign that would have biased the price toward the loser. In equilibrium, prediction market prices reflect the beliefs of the marginal informed trader, regardless of the demographic composition of the trader pool. This is the marginal-trader hypothesis of \cite{forsythe-nelson-neumann-wright1992}, which was proposed to explain the accuracy of the Iowa Electronic Markets. Our analysis is the first direct confirmation of this hypothesis in a major election market.

There are two qualifications associated with this observation. The first concerns what the composition evidence can and cannot establish. It establishes that the composition of Kalshi's traders did not transmit to Kalshi's price through the direction of their trades. Weighted by dollar volume, the direction of trades among traders with more than 100 trades does not differ significantly by partisanship, sex, or age, and most of that volume comes from two-sided traders whose positions show no partisan gap. It does not establish where the price came from. Kalshi's price tracked the prices of Polymarket, the Betfair exchange, and the bookmakers closely (Section~\ref{section:accuracy}), and our data, which covers one exchange, cannot separate price discovery on Kalshi from the incorporation of prices set elsewhere. The agreement among the four market-based forecasters in Section~\ref{section:accuracy} is therefore not four independent observations, and the trader bases of the other three are unobserved. For the paper's question this does not matter. Whichever venue led, the price on each was held in place by profit-seeking capital, and a trader who keeps one exchange's price in line with public prices elsewhere is acting on the profit motive alone. For a claim about Kalshi's price in particular, the evidence shows that Kalshi's composition did not distort it. It does not identify the venue on which the price was first discovered.

The second is a set of scope limitations. The Kalshi evidence covers one election observed over its final month, the period when forecasts of every kind are most accurate, and the outcome is known, so the profit results carry an ex post component that the calibration scores only partly remove. Sex is imputed. The partisanship measures are ecological, zip-code aggregates rather than individual registration. The absence of a partisan gap among heavy traders also holds in zip codes that voted more than 75\% or less than 25\% Republican, and the PredictIt registration data shows that the zip code vote share tracks registered party closely (Appendix~\ref{app:robust}). Direction is measured from buy-side flows. The PredictIt replication, with individual registration and donor status, addresses the ecological concern in a second market and finds the same qualitative pattern, and the historical record extends the accuracy pattern but supplies no trader-level behavior. Each limitation narrows a claim, and none affects the core finding, the decline of every partisan and demographic difference in the direction of trades as trade intensity rises, which is measured within a single platform, at the individual level, under consistent definitions.

\section{Prediction Markets in Historical Context}\label{section:historical}

Election betting has been a prominent feature of American politics since at least the 1880s. \cite{rhode-strumpf2004, rhode-strumpf2013} document that organized betting on presidential elections was conducted openly on the curb exchanges and in the hotels of New York from the 1884 election through the early twentieth century, and that the markets declined by World War II. The largest market operated on the Curb Market just outside the New York Stock Exchange, where trading in equities gave way to political futures in the weeks before an election, and hundreds of millions of current-day dollars changed hands. Smaller cities ran their own markets, and contracts covered national, state, and local races. These markets differed from their modern descendants in ways that could shape who participated. Trading was physical rather than digital. There was no order book, so a bettor either transacted at a quoted price or walked away. Prominent political operatives, including the Tammany Hall boss Richard Croker, bet openly. Above all, the markets were not anonymous. Newspapers covered the betting extensively and reported the names, stakes, and odds of large bettors. Figure~\ref{fig:nyt1900} reproduces a \textit{New York Times} report from the day before the 1900 election that lists bettors by name alongside the prevailing odds, among them Croker, a municipal assemblyman, the secretary to the incoming New York governor, a railroad lobbyist, and two bookmakers, with individual wagers of \$10,000, roughly \$400,000 in 2026 dollars.\footnote{The legal status of the markets was unsettled, but no bettor or market was prosecuted. New York courts treated casual wagers between private individuals as legal but would not enforce betting disputes, and an election bet formally disqualified the bettor from voting.}

\begin{figure}[ht]
  \centering
  \includegraphics[width=0.9\textwidth]{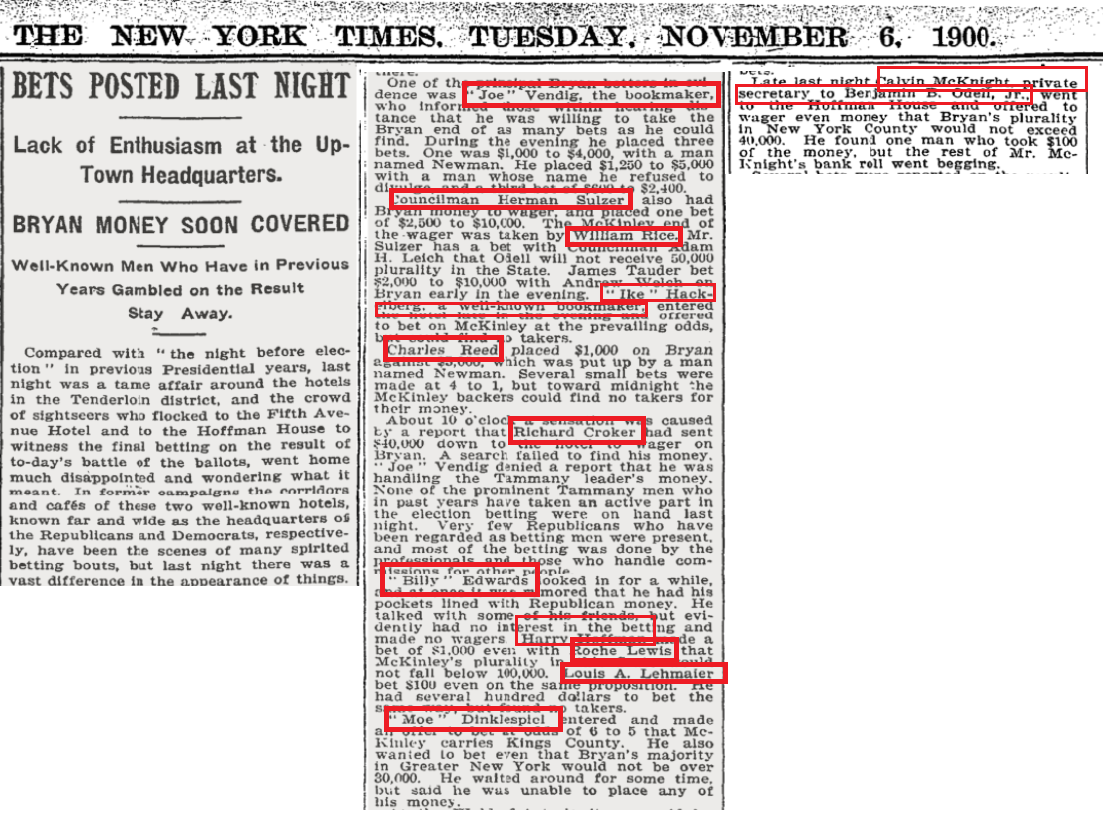}
  \caption{\textit{New York Times}, November 6, 1900. Boxed passages report the odds and name individual bettors with their stakes.}
  \label{fig:nyt1900}
\end{figure}

The published names make it possible to learn more about the demographics and socioeconomic status of the bettors listed.  We match the bettors named in New York newspaper accounts of the 1900 election to their records in the 1900 Census, yielding 80 matched bettors.\footnote{Census records are released after 72 years, so every Census through 1950 is available. Most of the 1890 population schedules were destroyed in a 1921 fire, which makes 1900, both a Census year and the start of the markets' peak period, the natural starting point. The 1900 election was a decisive re-election of William McKinley over William Jennings Bryan. Category denominators in Table~\ref{tab:hist1900} vary because not every Census field is populated for every matched bettor.} Table~\ref{tab:hist1900} and Figure~\ref{fig:hist1900} compare the matched bettors with the population of New York City in the same Census. Every bettor was male, and every bettor was white, which matched the city's male population (98\% white) and is the one dimension on which bettors and the city do not differ. On every other dimension they differ. The largest difference is in occupation, the Census's only proxy for income before 1940. Bankers or brokers were 44\% of bettors against 0.6\% of New York men, and 86\% were bankers, brokers, clerks, merchants, officials, manufacturers, or hoteliers, against 19\% of the city. Bettors were older, with a median age of 41 against 22 for New York men (32 among men aged 15 and over), and more likely to be married (77\% against 60\% of men aged 20 and over). They were more likely to be native born (74\% against 63\%), more likely to have a native-born father (41\% against 23\%), and, among the native born, far more likely to have been born outside New York State (28\% against 13\%). Among the foreign born, bettors came disproportionately from England and Canada (40\% of foreign-born bettors against 7\% of the city's foreign born) and not at all from Italy or Russia, which together supplied nearly a quarter of the city's immigrants. Chi-square tests of homogeneity reject equality with the city's distribution at the 5\% level on every dimension except race.

\begin{table}[p]
\centering
\caption{New York City election bettors in 1900 versus the New York City population, 1900 Census.}
\label{tab:hist1900}
\footnotesize\singlespacing
\renewcommand{\arraystretch}{0.95}
\begin{tabular}{lrr}
\toprule
 & Matched bettors (\%) & New York City (\%) \\
\midrule
\multicolumn{3}{l}{Occupation, males 10 and over: $\chi^2(6)=2{,}300$, $p<0.001$} \\
\quad Bankers and brokers & 44.4 & 0.6 \\
\quad Bookkeepers and clerks & 16.7 & 9.4 \\
\quad Merchants & 9.7 & 6.2 \\
\quad Officials & 5.6 & 0.4 \\
\quad Manufacturers & 5.6 & 1.9 \\
\quad Hotel keepers & 4.2 & 0.3 \\
\quad All other occupations & 13.9 & 81.2 \\
\midrule
\multicolumn{3}{l}{Age, males: $\chi^2(9)=70.7$, $p<0.001$} \\
\quad Under 15 & 0.0 & 30.9 \\
\quad 15--19 & 1.3 & 8.2 \\
\quad 20--24 & 11.4 & 9.5 \\
\quad 25--29 & 6.3 & 10.5 \\
\quad 30--34 & 12.7 & 9.7 \\
\quad 35--44 & 31.6 & 15.3 \\
\quad 45--54 & 19.0 & 8.6 \\
\quad 55--64 & 13.9 & 4.6 \\
\quad 65 and over & 3.8 & 2.5 \\
\quad Median age (years) & 41 & 22 \\
\midrule
\multicolumn{3}{l}{Marital status, males 20 and over: $\chi^2(2)=9.0$, $p=0.011$} \\
\quad Married & 76.9 & 60.4 \\
\quad Not married & 23.1 & 39.3 \\
\midrule
\multicolumn{3}{l}{Race, males: $\chi^2(2)=1.6$, $p=0.457$} \\
\quad White & 100.0 & 98.0 \\
\quad Black & 0.0 & 1.6 \\
\quad Other & 0.0 & 0.4 \\
\midrule
\multicolumn{3}{l}{Nativity, males: $\chi^2(1)=4.6$, $p=0.032$} \\
\quad Native born & 74.4 & 62.6 \\
\quad Foreign born & 25.6 & 37.4 \\
\midrule
\multicolumn{3}{l}{Father's nativity, all persons: $\chi^2(1)=13.4$, $p<0.001$} \\
\quad Native born & 40.8 & 23.1 \\
\quad Foreign born & 59.2 & 76.9 \\
\midrule
\multicolumn{3}{l}{State of birth among the native born, all persons: $\chi^2(12)=277.9$, $p<0.001$} \\
\quad New York & 72.4 & 87.3 \\
\quad Other states & 27.6 & 12.7 \\
\midrule
\multicolumn{3}{l}{Country of birth among the foreign born, all persons: $\chi^2(8)=38.2$, $p<0.001$} \\
\quad England & 30.0 & 5.4 \\
\quad Canada & 10.0 & 1.7 \\
\quad Germany & 20.0 & 25.4 \\
\quad Ireland & 15.0 & 21.7 \\
\quad Italy & 0.0 & 11.5 \\
\quad Russia & 0.0 & 12.2 \\
\quad Other countries & 25.0 & 22.2 \\
\bottomrule
\end{tabular}
\begin{minipage}{0.92\textwidth}
\vspace{0.4em}
\scriptsize
Notes: 80 bettors named in New York newspaper accounts of the November 1900 election and matched to the 1900 Census. All were male and all were white. The comparison population is New York City males in the 1900 Census except where the header says all persons, where the published Census tables do not separate the sexes. Median age for the city takes the midpoint of the median five-year interval. It is 32 among males aged 15 and over. Each chi-square statistic tests homogeneity of the bettor and city distributions over the categories shown, with age unknown, marital status unknown, and the detailed state and country lists included as in the source tables. Sources: 1900 Census, Volumes I and II (population by sex, age, conjugal condition, nativity, parentage, and occupation). Bettor matching by the authors.
\end{minipage}
\end{table}

\begin{figure}[ht]
  \centering
  \includegraphics[width=0.98\textwidth]{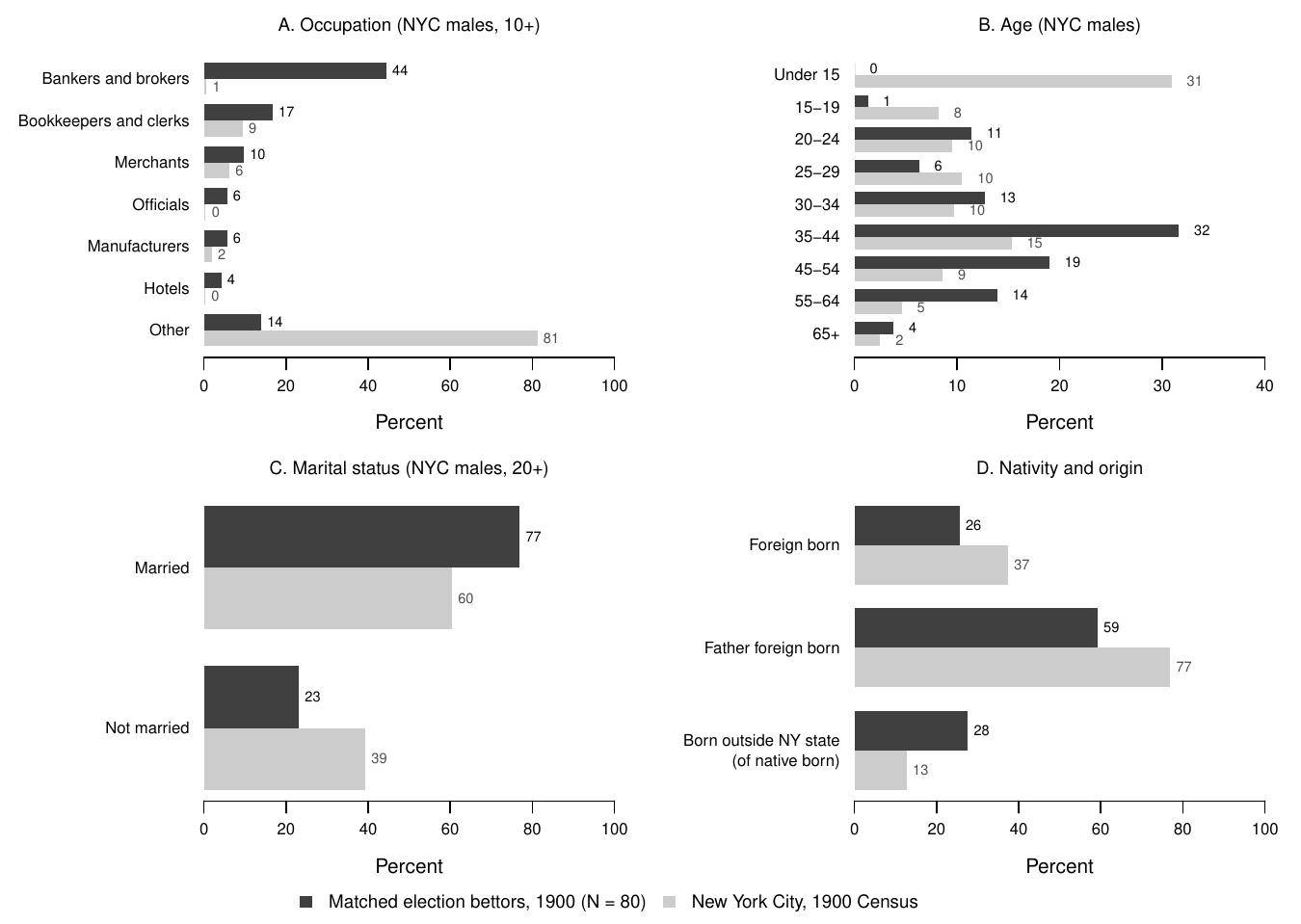}
  \caption{Matched election bettors of 1900 (dark bars, $N=80$) versus the New York City population in the 1900 Census (light bars). Panel D reports the foreign-born share, the share with a foreign-born father, and, among the native born, the share born outside New York State.}
  \label{fig:hist1900}
\end{figure}

Kalshi's traders are much younger than the electorate (Section~\ref{section:demographics}), whereas the 1900 bettors were much older than their city. Both markets were non-representative, in opposite directions. The markets driven by both groups were accurate. The particular way in which a market's participants fail to resemble the public is therefore not what governs its accuracy. The two sets of traders share a financial profile, not a demographic one.\footnote{Money was not the only thing wagered on the 1900 election. Newspapers also reported non-monetary bets, in which the loser agreed to wheel the winner down Broadway in a wheelbarrow or a baby carriage, push a peanut up a hill, walk to a distant city, marry, or shave. Because nothing but pride was at stake, almost anyone could make such a bet, and the people reported making them were drawn largely from the groups absent from the money markets. Participation in the money markets, not interest in the election, was what selected the narrow group who set prices.}
\begin{defbox}
In 1900 as in 2024, the people who set prices were those with capital to commit and a financial incentive to commit it prudently, regardless of their personal preferences.
\end{defbox}

The historical markets were also accurate, and accurate in an era without scientific polls (the first representative-sample polls were conducted in 1935) and with slow-moving news. \cite{rhode-strumpf2004} show that the New York betting markets correctly predicted the winner of every presidential election from 1884 to 1940 with one exception, that their implied probabilities were well calibrated, that lopsided winners were identified well before election day, that higher prices went with larger vote margins in national and local races alike, and that prices satisfied standard no-arbitrage conditions across the several New York venues and across cities.

The historical evidence supports our central finding. From the curb exchanges of Gilded Age New York to the regulated and blockchain platforms of the 2020s, the traders in both eras lacked demographic diversity, and the markets in both eras were accurate. 
 The 1900 evidence also removes several features of modern markets from the list of candidate explanations for their accuracy. Digital trading, anonymity, instantaneous news, and a young trader base were all absent in 1900, and the markets were accurate nonetheless. The pattern holds across very different market structures, regulatory environments, and electorates. 

\section{Conclusion}\label{section:conclusion}

This paper examines the relationship between demographic diversity and forecasting accuracy in political prediction markets. The paper has three findings.

First, the premise of the diversity critique is correct. Kalshi's trader base was predominantly young and male, a demographic profile far from that of the electorate whose behavior the market was forecasting.

Second, the predicted consequence never materialized. Kalshi's prices, like those of Polymarket, the Betfair exchange, and the bookmakers, were more accurate than every poll-based forecaster in the national race. Kalshi was more accurate than every poll-based forecaster in the battleground states and Senate races, and less accurate than FiveThirtyEight only in the thinly traded House districts. No poll-based model would have profited by trading against Kalshi's prices in any presidential contest group, and the markets correctly priced Trump as the favorite while poll-based models treated the race as a coin flip.

Third, the mechanism that reconciles these facts is financial incentives. The cohort of traders whose behavior most resembles political expression was 73\% long Harris and had the lowest returns. The cohort that dominates dollar volume was 64\% long Trump, had the highest returns, and, through the two-sided traders who supply most of its volume, bought at prices that proved on average fair. Under fractional partisan attribution by the home zip code's 2024 vote, the partisan gap in dollar-weighted trade direction falls from roughly seven percentage points among casual traders to zero among the traders with more than 100 trades who set prices, and the partisan gap in profitability declines in the same way. Trading profit is unrelated to sex, and among traders with more than 100 trades the direction of trades does not differ significantly by sex or age. The association of profit with geography and with age disappears once the side the trader took is accounted for. Formally, the price distortion that selected participation can produce is the product of a selection gap and the dependence of beliefs on type at the dollar-weighted margin. The gap is large, the association between politics and the direction of trades is not distinguishable from zero where volume concentrates, and on geography the gap has the opposite sign from the one the critics assume. The data is inconsistent with a model in which traders buy the outcome they prefer and consistent with one in which the market price reflects the beliefs of the marginal informed trader, the mechanism \cite{forsythe-nelson-neumann-wright1992} conjectured for the Iowa Electronic Markets and which our evidence confirms directly.

These findings challenge the widely held view, popularized by \cite{surowiecki2004}, that demographic diversity is necessary for accurate collective judgment.
\begin{defbox}
In prediction markets, demographic diversity is overruled by financial incentives. More generally, our analysis supports the claim that the structure of the institution in which a group operates matters more than the composition of the group.
\end{defbox}
The findings also answer the most common criticism leveled at prediction markets in 2024. Critics who inferred biased prices from an unrepresentative trader base were implicitly assuming that traders' positions follow their identities. The evidence rejects that assumption, and a century of evidence from demographically narrow, persistently accurate election markets suggests that the market mechanism does not require representative participants. A regulated exchange open to the public, in its first month of election trading, produced prices that were more accurate than every poll-based forecast and that were not noticeably biased by the inferred political preferences of its traders.

\appendix
\section{A Model of Selected Participation}\label{app:model}

This appendix sets out the model summarized in Section~\ref{subsec:framework}. Its purpose is to state when a prediction market whose participants are a systematically nonrepresentative subset of voters can nevertheless forecast the election, and to separate two notions of forecasting success that the empirical sections keep apart. The mapping from market prices to average beliefs depends in general on market structure, wealth, and risk preferences \citep{wolfers-zitzewitz2004, manski2006, wolfers-zitzewitz2006, he-treich2017}. The model here isolates the distinct issue of selected participation under a transparent aggregation benchmark.

\subsection{Environment}

There are $N$ voters. Voter $i$ has a type $v_i \in \mathbb{R}$, drawn independently from a distribution $F_V$ with finite second moment. Let $e$ denote an aggregate electoral state. Voter $i$ votes for candidate A when
\[
Y_i = \mathbbm{1}\{v_i + \beta(v_i)\,e + \varepsilon_i \ge 0\},
\]
where $\varepsilon_i$ is an idiosyncratic shock with $G(x) = P(\varepsilon_i \ge -x)$, so that $P(Y_i = 1 \mid v_i, e) = G(v_i + \beta(v_i)e)$. The expected vote share in state $e$ is $S(e) = \int G(v + \beta(v)e)\,dF_V(v)$, assumed continuous and strictly increasing in $e$, with a unique $e^*$ at which $S(e^*) = 1/2$. The realized vote share is $S_N = N^{-1}\sum_i Y_i$ and the realized winner is $W_N = \mathbbm{1}\{S_N > 1/2\}$.

\subsection{Selected participation}

Let $T_i \in \{0,1\}$ indicate that voter $i$ trades in the market, with $P(T_i = 1 \mid v_i = v) = \pi(v) \in (0,1]$. Participation may depend on type in any way. The distribution of types among traders is
\[
dF_T(v) = \frac{\pi(v)\,dF_V(v)}{E[\pi(v)]},
\]
and the \emph{selection gap} is $\Delta_T = E[v \mid T = 1] - E[v]$. Section~\ref{section:demographics} documents that $\Delta_T$ is large on every observed dimension of type.

\subsection{Market aggregation}

A binary contract pays one dollar if A wins. Let $p_N$ be its price and $q_i$ trader $i$'s subjective probability that A wins. Under a linear-demand benchmark, trader $i$ demands $x_i = \lambda_i(q_i - p_N)$ with $\lambda_i > 0$, and market clearing, $\sum_{i: T_i = 1} \lambda_i(q_i - p_N) = 0$, gives
\begin{equation}\label{eq:price}
p_N = \frac{\sum_{i: T_i=1} \lambda_i\, q_i}{\sum_{i: T_i=1} \lambda_i}.
\end{equation}
The price is the financially weighted average belief. With unequal stakes the relevant distribution of types is the dollar-weighted one, $dF_{\$}(v) \propto \lambda(v)\pi(v)\,dF_V(v)$, and the financially weighted selection gap $\Delta_{\$}$ is the selection gap under $F_{\$}$. Equation~\eqref{eq:price} is a benchmark rather than a universal equilibrium relationship. With heterogeneous budgets or risk preferences the price need not equal the simple average belief \citep{manski2006, he-treich2017}. Nothing below depends on the linear form beyond the price being a positively weighted average of beliefs.

\subsection{Forecast objects}

Let $\mathcal{I}_N$ be the information set on which the market forecast is conditioned immediately before the outcome is realized, and let $q^*_N = P(W_N = 1 \mid \mathcal{I}_N)$ be the true conditional probability of victory. We distinguish two notions of success. The market is \emph{probability consistent} if $p_N - q^*_N \to_p 0$. Writing $\widehat{W}_N = \mathbbm{1}\{p_N > 1/2\}$ for the market-implied winner, the market is \emph{winner consistent} if $P(\widehat{W}_N = W_N) \to 1$. Define the financially weighted belief error
\begin{equation}\label{eq:B}
B_N = \frac{\sum_{i: T_i=1} \lambda_i\,(q_i - q^*_N)}{\sum_{i: T_i=1} \lambda_i},
\end{equation}
so that under \eqref{eq:price}, $p_N - q^*_N = B_N$ identically.

\subsection{Belief formation and the source of price distortion}

The following results treat four mutually exclusive accounts of how traders form $q_i$ and identify in each the term that must vanish for the price to be unbiased. None of the conditions is that $F_T = F_V$.

\begin{proposition}[Forecasting from own type]\label{prop:owntype}
Suppose $q_i = h(v_i)$ and stakes are equal. Let $\bar q_V = E[h(v)]$ be the average belief a representative trader population would hold. Then
\[
p_N - \bar q_V = \frac{\mathrm{Cov}\big(\pi(v), h(v)\big)}{E[\pi(v)]}.
\]
\end{proposition}

\noindent Selected participation distorts the price only through the covariance between participation and the belief that type induces. With unequal stakes $\pi$ is replaced by $\lambda\pi$. Section~\ref{subsec:geo} estimates this covariance at the zip-code level and finds it near zero. Weighting zip codes by residents rather than by traders moves the long-Republican share of trading by under two percentage points.

\begin{proposition}[Representativeness is not necessary]\label{prop:correct}
If every trader holds the correct belief, $q_i = q^*_N$, then $p_N = q^*_N$ for any participation function $\pi$ and any positive weights $\lambda_i$.
\end{proposition}

\begin{proposition}[Inference from the trader population]\label{prop:pool}
Suppose traders treat the trader distribution $F_T$ as the electorate and forecast the vote share as $\widehat S_T(e) = \int G(v + \beta(v)e)\,dF_T(v)$. Then
\[
\widehat S_T(e) - S(e) = \frac{\mathrm{Cov}\big(\pi(v), G(v + \beta(v)e)\big)}{E[\pi(v)]}.
\]
\end{proposition}

\noindent The distortion vanishes when the two distributions imply the same election-relevant functional, which can hold when they differ on dimensions that do not alter it. Representativeness is sufficient for this but not necessary.

\begin{proposition}[Own-type projection]\label{prop:projection}
Let $\mu = E[v]$ and suppose trader $i$ perceives the mean type as $\widehat\mu_i = \mu + \rho(v_i - \mu) + \eta_i$ with $E[\eta_i \mid v_i] = 0$, and forms $q_i = Q(\widehat\mu_i)$ for a smooth $Q$. To first order, with equal stakes,
\[
p_N - Q(\mu) \simeq Q'(\mu)\,\rho\,\Delta_T,
\]
and with unequal stakes $\Delta_T$ is replaced by $\Delta_{\$}$.
\end{proposition}

\noindent This is the account that underlies the diversity critique. Traders see the electorate through their own type ($\rho > 0$) and the trader population is unrepresentative ($\Delta_T \neq 0$). The distortion is the product. It is zero if either factor is zero at the financially weighted margin, and its sign is the sign of $\rho\,\Delta_{\$}$. Sections~\ref{subsec:geo} and \ref{subsec:lpm} estimate $\rho$ by trade intensity. Section~\ref{section:demographics} reports $\Delta_{\$}$ on the geographic dimension and finds it negative.

\subsection{Probability consistency}

\begin{theorem}[Probability consistency without representativeness]\label{thm:prob}
Under the aggregation benchmark \eqref{eq:price}, the market is probability consistent if and only if $B_N \to_p 0$. The condition does not require $F_T = F_V$.
\end{theorem}
\begin{proof}
By \eqref{eq:price} and \eqref{eq:B}, $p_N - q^*_N = B_N$.
\end{proof}

Theorem~\ref{thm:prob} is a high-level condition. Individual belief errors must disappear after participation selection and financial weighting. Propositions~\ref{prop:owntype} to \ref{prop:projection} give its primitive form under each behavioral account. We make two remarks. First, a large market is not sufficient. If selection systematically shifts the dollar-weighted distribution of beliefs, the law of large numbers drives the price precisely to an incorrect probability. The price depends on the dollar-weighted belief distribution, not on the number of traders. Second, if $B_N \to_p B \neq 0$, the price converges to $q^*_N + B$ and probability consistency fails. The calibration score of Section~\ref{subsec:profit} is an observable counterpart of $B_N$ for each cohort.

\subsection{Winner consistency}

Probability consistency does not imply winner consistency. If $q^*_N$ remains bounded away from zero and one, the election retains irreducible uncertainty, and even a market with $p_N = q^*_N$ exactly names the realized winner with probability bounded below one. Winner consistency requires two further conditions.

\begin{theorem}[Winner consistency]\label{thm:winner}
Suppose (i) there is a deterministic sequence $a_N \downarrow 0$ with $P(|p_N - q^*_N| > a_N) \to 0$, (ii) $P(|q^*_N - 1/2| \le a_N) \to 0$, and (iii) $E[\min\{q^*_N, 1 - q^*_N\}] \to 0$. Then $P(\widehat W_N = W_N) \to 1$, whether or not $F_T = F_V$.
\end{theorem}
\begin{proof}
Let $W^*_N = \mathbbm{1}\{q^*_N > 1/2\}$. By (i) and (ii), $P(\widehat W_N \neq W^*_N) \le P(|p_N - q^*_N| > a_N) + P(|q^*_N - 1/2| \le a_N) \to 0$. By (iii), $P(W^*_N \neq W_N) = E[\min\{q^*_N, 1-q^*_N\}] \to 0$. The union bound gives the result.
\end{proof}

Condition (iii), outcome concentration, has a primitive foundation in a large electorate. If, conditional on $(e, v_1, \dots, v_N)$, the idiosyncratic shocks are independent with bounded variances, then $S_N - S(e) \to_p 0$. If in addition the limiting election is not at the knife's edge, $P(|S(e) - 1/2| > \kappa) \to 1$ for some $\kappa > 0$, and the market information set reveals $e$ or enough of it to determine $S(e)$, then $q^*_N \to_p \mathbbm{1}\{S(e) > 1/2\}$ and (iii) holds. Conversely, if $\liminf_N E[\min\{q^*_N, 1 - q^*_N\}] > 0$, no forecast measurable with respect to $\mathcal{I}_N$ can be winner consistent, because the Bayes classifier $W^*_N$ is the best such forecast and its error is exactly that expectation.

The distinction organizes the accuracy evidence of Section~\ref{section:accuracy}. That the markets named the winner in 2024 is one realization of $W_N$, and because the down-ballot contests share the aggregate state $e$, the battleground and Senate calls are largely the same realization. The Brier scores, and the trader-level calibration scores, speak to probability consistency, which is the property that selected participation could have violated and did not.

\subsection{Parametric illustration}

If $v \sim N(\mu, \sigma_v^2)$ and $\pi(v) \propto e^{\gamma v}$, then $\Delta_T = \gamma\sigma_v^2$ and, under Proposition~\ref{prop:projection}, $p_N - Q(\mu) \simeq Q'(\mu)\,\rho\,\gamma\,\sigma_v^2$. Selection strength $\gamma$, type dispersion $\sigma_v^2$, and projection $\rho$ enter multiplicatively. However strong the selection, the price is undistorted when the traders who carry the financial weight do not project.

\section{Robustness of the Partisan Attribution}\label{app:robust}

This appendix reports the three checks summarized in Section~\ref{subsec:geo}. The attribution treats a zip code's vote share as the probability that a resident trader is a Republican. Traders are a selected draw from their zip codes, so the measure could understate individual partisanship, and understate it unevenly if the home zip codes of heavy traders say less about their politics than those of casual traders. Three checks address the concern (Table~\ref{tab:homog}). First, restricting to homogeneous zip codes, those that voted more than 75\% or less than 25\% Republican, where residents are politically homogeneous, the raw gap in dollar-weighted direction between Red and Blue residents is 31 percentage points for one-trade users and 37 percentage points for the 2 to 10 cohort, several times the attribution gap, and it is $-2.5$ percentage points ($t=-0.4$) for the 101+ cohort. Attenuation compresses the casual-cohort gaps. It does not create the absence of a gap in the most active cohort. Second, the home zip codes of heavy traders are no less informative than those of casual traders. The Bernoulli variance of type given zip, $w(1-w)$, is 0.18 for the 101+ cohort's dollars against 0.20 to 0.21 for the other cohorts, and 58\% of the heavy cohort's volume comes from homogeneous zip codes against 25 to 33\% elsewhere. Third, the PredictIt data, which records registered party, shows that the zip vote share is a calibrated proxy. The PredictIt data records the first three digits of each account's zip code, which we call the zip-code prefix. Among the 46,739 PredictIt accounts registered Democrat or Republican, a regression of the Republican indicator on the 2024 Republican vote share of the account's zip-code prefix has slope 1.03 (standard error 0.01) and intercept $-0.02$. Among accounts that traded the slope is 0.93 (standard error 0.04). Among PredictIt traders in zip-code prefixes that voted more than 60\% Republican, the average zip-code prefix voted 67.8\% Republican and 62.4\% of the traders registered with a major party are registered Republican. Among traders in zip-code prefixes that voted less than 40\% Republican, the average zip-code prefix voted 24.9\% Republican and 20.5\% are registered Republican. In both groups the share registered Republican is within six percentage points of the vote share.

\begin{table}[h!]
\centering
\caption{Attenuation check: direction of trades in homogeneous zip codes (2024 two-party Republican vote share above 75\%, ``Red'', or below 25\%, ``Blue''), by cohort. Direction is the share of classified election-market dollars committed long Republican. The raw gap compares Red-zip with Blue-zip residents directly. The attribution gaps apply the paper's fractional attribution to homogeneous-zip traders only and to all traders. Bootstrap $t$ over traders in parentheses.}
\label{tab:homog}
\footnotesize
\begin{tabular}{lrrrrrrr}
\toprule
 & \multicolumn{2}{c}{Users} & \multicolumn{2}{c}{\% long Republican} & Raw gap & Attribution gap & Attribution gap \\
Cohort & Red & Blue & Red & Blue & Red $-$ Blue & homogeneous zips & all traders \\
\midrule
=1      & 1,158 & 5,104  & 46.5 & 15.1 & $31.4$ ($7.6$)  & $12.9$ ($8.2$)  & $6.8$ ($10.0$) \\
2--10   & 6,708 & 18,335 & 75.7 & 38.6 & $37.2$ ($17.1$) & $17.4$ ($15.4$) & $6.9$ ($7.8$) \\
11--100 & 2,682 & 8,866  & 51.2 & 44.2 & $7.1$ ($1.1$)   & $4.5$ ($2.1$)   & $4.5$ ($4.3$) \\
101+    & 441   & 2,417  & 61.9 & 64.3 & $-2.5$ ($-0.4$) & $-1.3$ ($-0.7$) & $-0.2$ ($-0.1$) \\
\bottomrule
\end{tabular}
\end{table}

\bibliography{references}

\end{spacing}
\end{document}